\documentclass[aps,prx,twocolumn,showpacs,superscriptaddress,preprintnumbers,longbibliography]{revtex4-2} 
\usepackage{dcolumn}   
\usepackage{bbm}        
\usepackage{amssymb}   
\usepackage{amsmath}

\usepackage[colorlinks=true,citecolor=blue,linkcolor=red]{hyperref}
\usepackage{cleveref}
\usepackage[normalem]{ulem}
\usepackage{comment,braket,physics}
\usepackage{color,xcolor}
\usepackage{slashed}
\usepackage{braket}
\usepackage{amssymb,amsmath,color}
\usepackage[caption=false]{subfig}
\usepackage{graphicx}  
\usepackage[export]{adjustbox}
\usepackage{multirow}
\usepackage{rotating,booktabs}
\usepackage[verbose]{placeins}
\usepackage{mathrsfs}
\usepackage{stackengine}
\usepackage{soul}
\usepackage{float}

\newcommand{\Eq}[1]{Eq.~(\ref{#1})}

\newcommand{\Eqs}[2]{Eqs.~(\ref{#1}-\ref{#2})}
\newcommand{\Fig}[1]{Fig.~\ref{#1}}

\renewcommand{\mel}[3]{\langle #1 | #2 | #3 \rangle}
\newcommand{\CC}{{\mathbb{C}}}
\newcommand{\EE}{{\mathbb{E}}}
\newcommand{\RR}{{\mathbb{R}}}

\newcommand{\poly}{{\mathrm{poly}}}

\newtheorem{lemma}{Lemma}
\crefname{equation}{Eq.}{Eqs.}
\Crefname{equation}{Equation}{Equations}

\crefname{figure}{Fig.}{Figs.}
\Crefname{figure}{Figure}{Figures}

\crefname{table}{Tab.}{Tabs.}
\Crefname{table}{Table}{Tables}

\crefname{appendix}{App.}{Apps.}
\Crefname{appendix}{Appendix}{Appendices}

\makeatletter
\def\blfootnote{\xdef\@thefnmark{}\@footnotetext}
\makeatother
\begin{document}
\preprint{IQuS@UW-21-129}

\title{Benchmarking quantum simulation at scale}

\author{Jeremy Hartse}
\thanks{These authors contributed equally to this work (alphabetical order). \url{jhartse@uw.edu}, \url{mraza98@unm.edu}, \url{shravan@unm.edu}}
\affiliation{InQubator for Quantum Simulation (IQuS), Department of Physics, University of Washington, Seattle, WA 98195, USA.}

\author{Mohsin Raza}
\thanks{These authors contributed equally to this work (alphabetical order). \url{jhartse@uw.edu}, \url{mraza98@unm.edu}, \url{shravan@unm.edu}}
\affiliation{Center for Quantum Information and Control, University of New Mexico, Albuquerque, NM 87106, USA}
\affiliation{Department of Physics and Astronomy, University of New Mexico, Albuquerque, NM 87106, USA}

\author{Shravan Shravan}
\thanks{These authors contributed equally to this work (alphabetical order). \url{jhartse@uw.edu}, \url{mraza98@unm.edu}, \url{shravan@unm.edu}}
\affiliation{Center for Quantum Information and Control, University of New Mexico, Albuquerque, NM 87106, USA}
\affiliation{Department of Physics and Astronomy, University of New Mexico, Albuquerque, NM 87106, USA}

\author{Ivan H. Deutsch}
\affiliation{Center for Quantum Information and Control, University of New Mexico, Albuquerque, NM 87106, USA}
\affiliation{Department of Physics and Astronomy, University of New Mexico, Albuquerque, NM 87106, USA}

\author{Niklas Mueller}
\affiliation{Center for Quantum Information and Control, University of New Mexico, Albuquerque, NM 87106, USA}
\affiliation{Department of Physics and Astronomy, University of New Mexico, Albuquerque, NM 87106, USA}

\begin{abstract}
The applications for which quantum computers will clearly outperform classical computers are still being identified and benchmarking such an advantage is challenging. We propose a scalable verification scheme for non-equilibrium quantum simulation based on stabilizer scars, a special class of quantum many-body scars, whose structure ensures both classical simulability and efficient direct fidelity estimation. Assuming a physically motivated error model, we show that the fidelity of quantum simulating these states bounds the fidelity of classically intractable simulations, providing a benchmark for quantum-advantage experiments in non-equilibrium dynamics. 
\end{abstract}
\maketitle

{\em Introduction.}
The motivation for quantum computers is to address problems whose complexity lies beyond the reach of classical methods~\cite{preskill2012quantum,aaronson2016complexity,harrow2017quantum,deutsch2020harnessing}. Target areas include cryptography~\cite{shor1994algorithms}, optimization~\cite{abbas2024challenges}, sampling~\cite{shepherd2009temporally,aaronson2011computational,bremner2017achieving,boixo2018characterizing}, and quantum simulation  ranging from chemistry and materials to high energy and nuclear physics~\cite{feynman2018simulating,lloyd1996universal,cirac2012goals,gross2017quantum,beck2019nuclear,cao2019quantum,altman2021quantum,bauer2023quantum,bauer2023quantuma}.

Quantum simulation offers a compelling route to quantum advantage, yet establishing it rigorously remains challenging~\cite{kitaev2002classical,hangleiter2017direct,fitzsimons2018post,elben2020cross,eisert2020quantum,carrasco2021theoretical,mills2021application,kliesch2021theory,aaronson2024verifiable,proctor2025benchmarking}.  Apart from experimental challenges, numerous classical algorithms already efficiently describe properties of many quantum systems~\cite{kolb2004hydrodynamic,schollwock2011density,cirac2009renormalization,probert2011electronic}. Moreover, noise can render quantum problems classically simulable~\cite{aharonov2000quantum,harrow2003robustness,flannigan2022propagation,schuster2025polynomial,angrisani2025simulating,shravan2026efficientsimulationnoisyiqp} and the issue is not only whether a computation is correct, but \textit{how} correct it is.  Finally, verifying a classically infeasible computation  may be as difficult as simulating the computation itself. Verifying the correctness and advantage in analog quantum simulators is an even greater challenge~\cite{trivedi2024quantum,liu2025efficiently}.

Early quantum supremacy demonstrations have been based on the computational complexity of sampling from certain hard distributions~\cite{arute2019quantum,decross2025computational,gao2025establishing,morvan2024phase,wu2021strong}.  Benchmarking such {unstructured problems}, however, is challenging: the very property that makes classical emulation difficult also renders verification exponentially costly. Problem agnostic benchmarks may be difficult to scale~\cite{emerson2005scalable,knill2008randomized,proctor2019direct,cross2019validating,nielsen2021gate,helsen2022general,BlumeKohout2020volumetricframework}. In contrast, we argue that quantum simulation applications could provide extra structure for verification.

We propose a method for verifying large-scale quantum simulations of nonequilibrium dynamics,  aiming at the near-term regime where a broad quantum advantage may become achievable. We make use of models hosting quantum many-body scars (QMBS)~\cite{bernien2017probing,turner2018weak,serbyn2021quantum,moudgalya2022quantum,chandran2023quantum},  related to lattice gauge theory (LGT) models, which are key in many domains~\cite{halimeh2025quantum,fradkin2013field,kleinert1989gauge,gottesman2010introduction,sarma2006topological,nayak2008non}. 

The classical simulability of QMBS is not a new idea~\cite{bernien2017probing,turner2018weak}. Our approach leverages the structure of recently proposed \textit{stabilizer scars}~\cite{hartse2025stabilizer,dooley2026parent,hokkyo2026exact,gupta2026exact}, enabling not only efficient classical computation, but also guaranteeing efficient benchmarking that constrains the fidelity when simulating non-classically simulable, non-scar dynamics. 

{\em QMBS-based benchmarking.} We consider Trotterized time evolution of a model with Hamiltonian $H$, represented by a quantum circuit $U$, and an input state $\ket{\psi(0)}$. The task is to verify the fidelity of the output of a circuit $\sigma = \mathcal{E}_{U} \left( \ketbra{\psi(0)}{\psi(0)} \right) $ computed on a device, relative to an (ideal) target state $\rho  = U \ketbra{\psi(0)}{\psi(0)}U^{\dagger}$, where $\mathcal{E}_{U}$ represents the  noisy version of the ideal unitary circuit $U$. The ideal circuit is neither a Clifford nor otherwise special circuit and is assumed, though not proven for the present application, to be classically hard for a generic input.

Our protocol is compactly summarized in~\Fig{fig:overview} with the following properties:
(i)~It is applied to models that host an exact, polynomially large QMBS subspace, spanned by stabilizer states, embedded in an exponentially large Hilbert space of typical states, with dynamics in the QMBS subspace classically tractable. (ii)~The fidelity $F(\rho, \sigma) $ of any state simulated within the QMBS subspace relative to the target $ \rho$ can be efficiently measured experimentally using direct fidelity estimation~\cite{flammia2011direct,da2011practical}, with both the number of shots and observables scaling polynomially with system size. (iii)~Typical (non-scar) inputs cannot be classically simulated or have their fidelity directly efficiently measured, however, with local depolarization noise, we show their fidelity  is nonetheless  as good as that of the scar states.

\begin{figure}[t]
    \centering
    \includegraphics[width=0.93\linewidth]{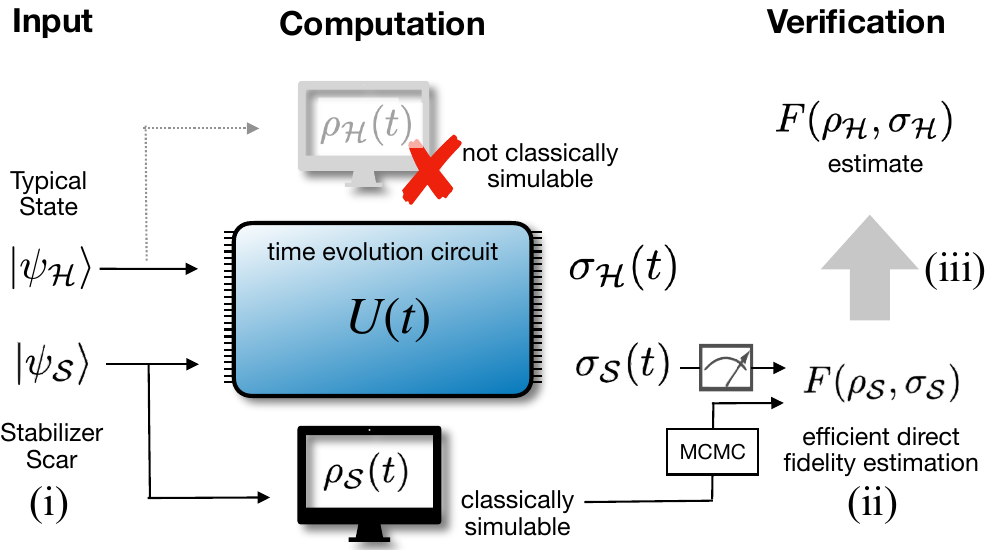}
    \caption{\textit{Benchmarking Quantum Simulations at Scale.}  A protocol for benchmarking the fidelity of quantum simulations of non-equilibrium dynamics in regimes beyond classical simulability. We consider models which are non-integrable, for which typical states, $|\psi_\mathcal{H}\rangle$, thermalize.
    Such dynamics are not efficiently classically simulable, and benchmarking performance for such states is  intractable.  For models that host stabilizer scars, under consideration here, we can exploit the fact that for initial states in this subspace, $|\psi_\mathcal{S}\rangle$ dynamics are classically simulable and we can perform direct fidelity estimation  efficiently. 
    Assuming local depolarizing noise, we demonstrate for a specific stabilizer scar model~\cite{hartse2025stabilizer} that the fidelity of classically non-simulable states is approximately equal to that of the states restricted to the scar subspace, thus providing a method to verify performance for arbitrary inputs.  
    \label{fig:overview}}
\end{figure}

{\em Efficient direct fidelity estimation in a stabilizer scar subspace.}  To prove (ii), we adapt the direct fidelity estimation (DFE) protocol of Ref.~\cite{flammia2011direct}.  The fidelity between  a pure target state $\rho$ and its noisy experimental realization $\sigma$ is
\begin{equation}\label{eq:def_fidelity}
    F( \rho, \sigma) = \text{Tr}(\rho \sigma) = \sum_{k = 1}^{d^{2}} \chi_{\rho}(k) \chi_{\sigma}(k),
\end{equation}
$\chi_{\rho}(k) = \text{Tr}[\rho W_{k}]/\sqrt{d}$, $d=2^n$ for $n$ qubits, and $W_{k}\in \mathcal{P}_{n} = \{\mathbb{I},X,Y,Z \}^{\otimes n}$ is a $n$-qubit Pauli observable; $X,Y,Z$ are Pauli matrices.  To estimate \Eq{eq:def_fidelity}, one draws, with probability $P_{\rho}(k)= \chi_{\rho}^2(k) $, samples $\{ k_i \}_{i=1}^s$
at random. One measures the estimator $\chi_\sigma(k_i)/\chi_\rho(k_i)$ where, in our case, $\chi_\sigma(k_i)$ is obtained from a Pauli measurement from experiment and $\chi_\rho(k_i)$ classically computed. The average of this estimator is guaranteed to be $\epsilon$-close to the true fidelity with probability $1-\delta$ if $s=\lceil 1/(\epsilon^{2} \delta ) \rceil$ if the expectation value $\chi_\sigma(k_i)$ is known exactly~\cite{flammia2011direct}. In a quantum simulation, however, $\chi_\sigma(k_i)$ is estimated from $m_i$ copies of the state, requiring 
\begin{align}\label{eq:boundDFE}
    m_{i} = \left\lceil \frac{2\log(2/\delta)}{|\mathrm{Tr}[\rho W_{k_{i}}]|^2 \, s \epsilon^{2} }\right\rceil\,.
\end{align} 
\Eq{eq:boundDFE} ultimately controls the efficiency of the protocol. 
For a typical state, the situation is hopeless, since $|\mathrm{Tr}[\rho W_{k_{i}}]|^2$ is exponentially small in the system size for all but very special observables. 

Consider now (i), a QMBS subspace of dimension $d_s$, polynomial in qubit number $n$, spanned by an orthonormal set of stabilizer basis states, $\{ \varphi_{\alpha}\} $. Let  $\rho \in \mathcal{S}$, where $\mathcal{S} = \text{span}(\{ \varphi_{\alpha}\})$, be pure. DFE is efficient for an arbitrary superposition $\rho\in \mathcal{S}$. To see this, we analyze the distribution of $|\mathrm{tr}[\rho W_{k_{i}}]|^2$ for such states by considering the stabilizer  R\'enyi entropy (SRE)~\cite{leone2022stabilizer} 
\begin{equation}\label{eq:def_stab_renyi_entropy}
    M_{\alpha} \equiv \frac{1}{1-\alpha} \log(\xi_\alpha)\,, \quad \xi_\alpha \equiv \frac{1}{d}\sum_{W_p\in \mathcal{P}_{n}} \text{Tr}[\rho W_{p}]^{2\alpha}\,.
\end{equation}
Exploiting a result of Haug and Piroli~\cite{haug2023stabilizer} [their Eq.~(11)] which  bounds the stabilizer fidelity using SREs, and as we show in SM, the stabilizer fidelity of the worst $\rho \in \mathcal{S}$ is lower bounded by $1/d_s$, we find that the $\alpha$-SRE
\begin{equation}\label{eq:stab_renyi_entropy_up_bound}
    M_{\alpha} \leq \frac{2\,\alpha}{\alpha-1}\log(d_{s}),\,  \quad \forall \alpha >1.
\end{equation}
 For a pure state $\rho$, \Eq{eq:def_stab_renyi_entropy} specifies a probability distribution, $P_\rho(k_i) $.  \Eq{eq:stab_renyi_entropy_up_bound}  implies a \textit{polynomial} lower bound on the average of this distribution,
\begin{align} \label{eq:lowerboundonexpenctationvaluemain}
& \xi_{2} =\underset{k_{i}\sim P_{\rho}}{\mathbb{E}}[\Tr[\rho W_{k_{i}}]^{2}] \geq \frac{1}{d_{s}^{4}}\,.
\end{align}
Since $\Tr[\rho W_{k_{i}}]^{2}$ is a bounded random-variable, it is sub-gaussian by Hoeffding's lemma, therefore, it has an exponential tail-bound:
\begin{equation}\label{eq:tail_bound}
    \text{Pr}[|\Tr[\rho W_{k_{i}}]^{2} - \underset{k_{i}\sim P_{\rho}}{\mathbb{E}}[\Tr[\rho W_{k_{i}}]^{2}]| \geq \gamma]\leq 2 e^{ -2\gamma^{2}},\,
\end{equation}
where $\gamma >0$  ~\cite{boucheron2013concentration}, thus  the distribution of Pauli expectations is concentrated around the inverse polynomial lower-bound; see Supplementary Material (SM) \ref{App:efficient_fidelity_estimation}.

A concrete bound on the total number of copies, $
N_{\rho} =\sum_{i=1}^{s}m_{i}$ can be obtained following Leone \textit{et al.}~\cite{leone2023nonstabilizerness}. Using an $\epsilon$-truncation of the state~\cite{flammia2011direct}, they derive an upper bound on $N_{\rho}$ in terms of the $0$-SRE. In $\mathcal{S}$, $M_0$ scales only as $\log(n^2)$ (more precisely, the set of nonzero expectation values is $d(n^2-n+2)$ for the model considered below). Consequently,
\begin{align}\label{eq:leonebound}
    N_{\rho}^{\text{\cite{leone2023nonstabilizerness}}} \leq \frac{16}{\epsilon^4}\log(2/\delta) \, \mathcal{O}(d_{s}^2)\,.
\end{align}
However, in practice no such $\epsilon$-truncation needs to be introduced: An observable with exponentially small expectation value may occur but, because of \Eqs{eq:lowerboundonexpenctationvaluemain}{eq:tail_bound}, a Markov Chain Monte Carlo (MCMC) algorithm for DFE (like the Metropolis-Hastings algorithm that we employ below)  generating samples from $P_\rho$ would accept those only with exponentially vanishing probability.
Regardless of the argumentation, DFE is \textit{efficient}: The fidelity between a state $\sigma$ and an ideal pure state $\rho$ can be efficiently estimated if $\rho\in \mathcal{S}$, the stabilizer scar subspace. 

{\em Fidelity estimation for classically non-simulable states.} 
We now argue~(iii), that under suitable assumptions on the noise, the fidelity estimate of a state in $\mathcal{S}$ also serves as a benchmark for the fidelity of non-classically computable non-scar states quantum computed with the same circuit.
We assume, for simplicity, that the dominant error  is approximated by local depolarizing noise with probability $p$, i.e., $\mathcal{E}(\rho) = \bigotimes_{i=1}^n \mathcal{E}^{(i)}(\rho) $ with  $\mathcal{E}^{(i)}(\rho)= (1-p)\rho + p\frac{\mathbb{I}}{2}$ 
 and consider the {\em average} fidelity, 
\begin{align}\label{eq:avfidel}
  F(\mu)\equiv   \mathbb{E}_{|\psi\rangle \sim \mu} \left[ \langle \psi | \mathcal{E}\big(|\psi\rangle \langle \psi|\big) | \psi \rangle \right],
\end{align}
where the average is over input states drawn  from a given ensemble $\mu$, either  the Haar measure over the full Hilbert space of $n$ qubits, $|\psi\rangle \sim \mu_{\mathcal{H}}$, or the uniform measure over the stabilizer scar subspace, $\ket{\psi}\sim \mu_{\mathcal{S}}$, with dimension $d_s=\poly(n)$. In the full Hilbert space, 
\begin{align}\label{eq:fidscar}
  F(\mu_{\mathcal{H}})= \frac{1}{1+1/d}\left[\frac{1}{d}+\left(1-\frac{3p}{4} \right)^n \right]\approx\left(1-\frac{3p}{4} \right)^n 
\end{align}
where, in the second relation, exponentially suppressed terms proportional to inverse powers of $d$ have been neglected for $p<1$ [$F(\mu_{\mathcal{H}})=1/d$ at $p=1$]; see SM~\ref{supp:fidelity} for details. 
In contrast, for input states uniformly drawn from the scar subspace,
\begin{align}\label{eq:scarsubspacefid}
 F(\mu_{\mathcal{S}})=\frac{1}{d_s(d_s+1)} \sum_{i=1}^{d^2}\left[ \text{Tr}[K_i^\mathcal{S} K_i^{\mathcal{S}\dagger}] + |\text{Tr}[K_i^\mathcal{S} ]|^2 \right],
\end{align}
where $K_i^\mathcal{S} \equiv \Pi_{\mathcal{S}} K_i \Pi_{\mathcal{S}}$ are the Kraus operators of the channel,  proportional to Pauli operators $W_i$, i.e., $K_i \in \{{(1-\frac{3p}{4})^{\frac{1}{2}}} \mathbb{I}, {(\frac{p}{4})^{\frac{1}{2}}}X, {(\frac{p}{4})^{\frac{1}{2}}}Y, {(\frac{p}{4})^{\frac{1}{2}}}Z \}^{\otimes n}$,  projected onto $\mathcal{S}$ via  $\Pi_{\mathcal{S}}$.
We now derive a quantitative relation between the fidelity in the scar subspace, \Eq{eq:scarsubspacefid}, and in the full Hilbert space, \Eq{eq:fidscar}. 

{\em Prototype model.} 
With knowledge of error model and the structure of the scar subspace, \Eq{eq:scarsubspacefid} can be tightly lower and upper bounded. The following   Ising model on $n=2L$ qubits posses a stabilizer scar subspace~\cite{hartse2025stabilizer}, see SM~\ref{supp:model}, and satisfies both (i) and (ii),
\begin{align}\label{eq:Hamiltonian}
    H=&-\sum_{p}X_p - g \sum_{p,{a}={1},{2}}  Z_{p}Z_{p+\hat{a}}\,,
\end{align}
where $p=(p_1,p_2)$ defined on a lattice with size $L\times 2$ and periodic boundary conditions (PBC), and subject to a parity constraint, $\prod_{p}X_p=1$ resulting in Hilbert space dimension $d = 2^{2L-1}\approx 2^{2L}$; $X_p$ and $Z_p$ are Pauli matrices. For odd $L$, the model hosts a QMBS subspace of dimension $d_s=2n=4L$, spanned by an orthonormal set of stabilizer basis states, $\varphi_{\alpha,k} = | \varphi_{\alpha,k} \rangle \langle \varphi_{\alpha,k} | $,
\begin{align}\label{eq:ss1}
    \varphi_{\alpha,k}\equiv \prod_{q=1}^L \frac{1+s_{q,\alpha}^{(k)}Z_q Z_{q+L}}{2}\frac{1+r_{q,\alpha}^{(k)}X_q X_{q+L}}{2}
\end{align}
where the notation is such that $q$ labels sites starting in the top and continuing in the bottom row, i.e., $q\equiv(q_1,1)$ and $q+L\equiv(q_1,2)$; and $r_{q,a}^{k}=-1$ for $k\neq q$ and $1$ otherwise, and
\begin{align}\label{eq:ss2}
s_{q,1}^{(k)} &= s_{q,4}^{(k)} = - s_{q,2}^{(k)} = - s_{q,3}^{(k)}
= \sigma_{q}^{(k)} (-1)^{k-q}\nonumber \text{ for } q\neq k\,,\\
s_{k,2}^{(k)} &= s_{k,4}^{(k)} = - s_{k,1}^{(k)} = - s_{k,3}^{(k)} = 1\text{ for } q= k.
\end{align}
where $\sigma_{q}^{(k)} = +1$ for $q<k$ and $-1$ for $q>k$.  \Eq{eq:Hamiltonian} is dual to $\mathbb{Z}_2$ LGT; see SM~\ref{supp:model} for details, or Ref.~\cite{hartse2025stabilizer}. The evolution of any state in the subspace and any expectation value can be efficiently classically computed. 

\begin{figure}[t]
    \centering
    \includegraphics[width=0.8\linewidth]{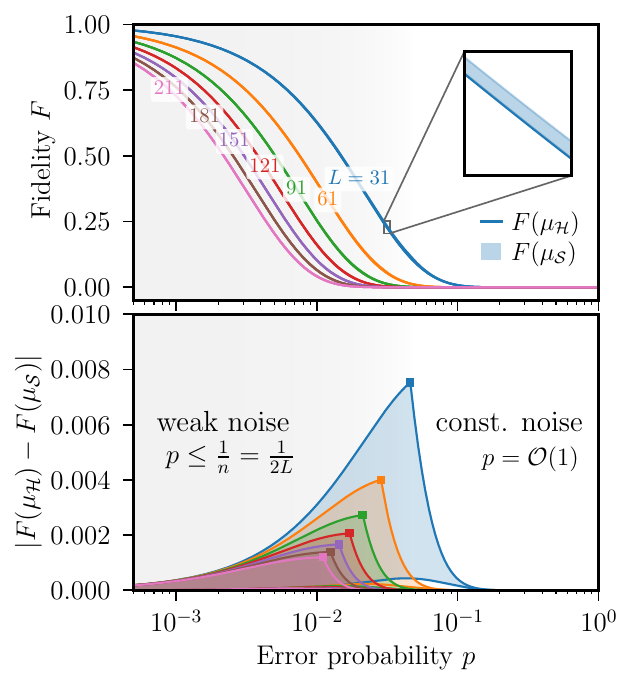}
    \caption{\textit{Noise acting on scar vs. non-scar states.} Comparison between the fidelity under a local depolarizing channel, averaged over Haar-random input from the full Hilbert space, $F({\mu_{\mathcal{H}}})$ (solid lines), and  bounds of the respective fidelity for scar states $F(\mu_{\mathcal{S}})$ (shaded bands, \Eq{eq:loosebound} and SM~\ref{supp:fidelity}). The difference, shown in the lower panel, vanishes as $n \to \infty$ or for small $p$, and (trivially) for $p\rightarrow 1$. The tighter bound is shown, and the crossover between where \Eq{eq:loosebound} is tighter and where the SM result applies is marked by square markers.}
    \label{fig:noisebounds}
\end{figure}
\begin{figure}[t]
    \centering
    \includegraphics[width=0.8\linewidth]{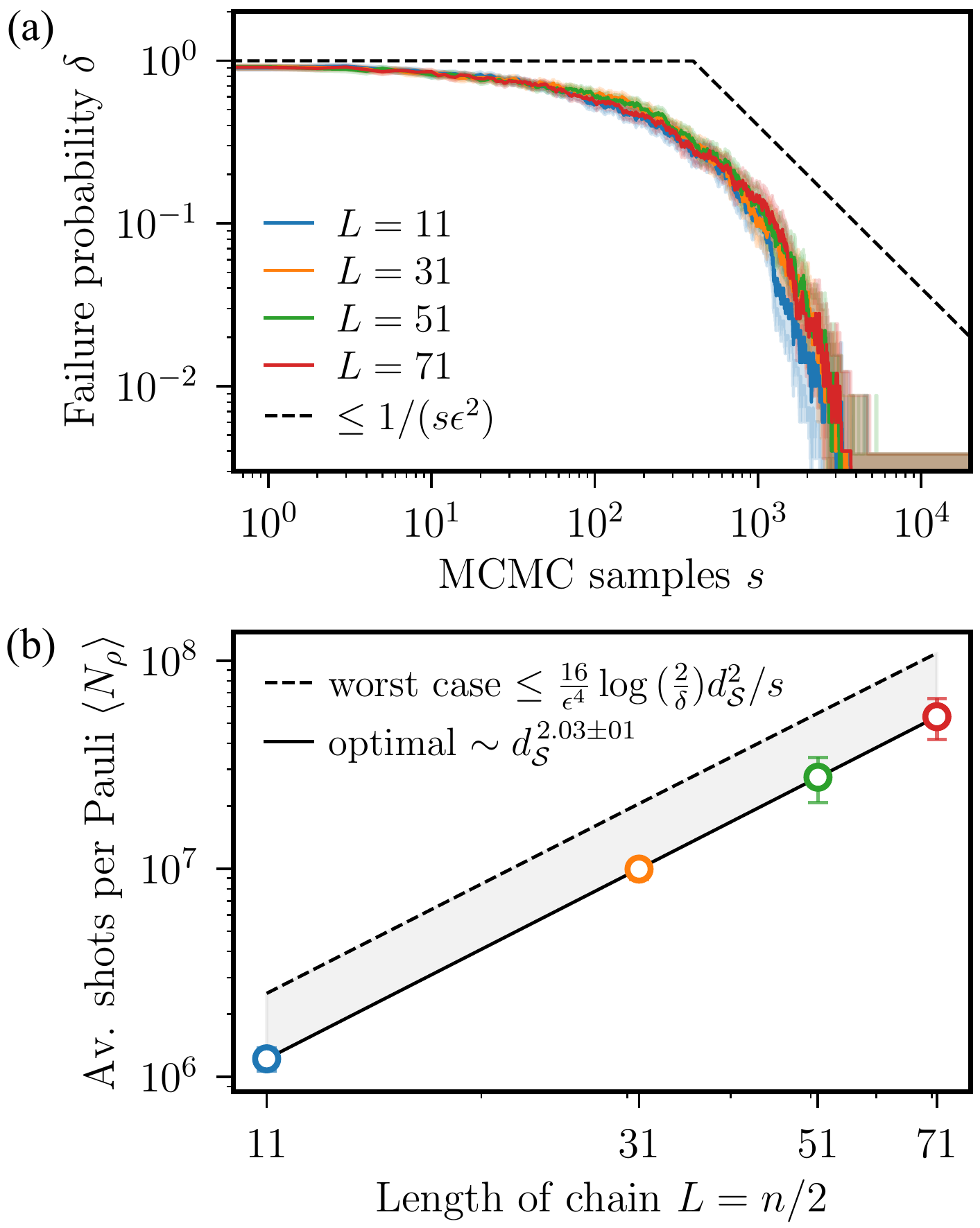}
    \caption{\textit{Numerical MCMC demonstration.}
(a) Failure probability to reach precision $\epsilon=5\cdot 10^{-2}$  of MCMC-based DFE between a target $\rho\in\mathcal{S}$ and a  reference state $\sigma\in\mathcal{S}$  as a function of the number of sampled Pauli operators with exact measurements.   Results shown for $L=11,31,51,71$ ($n=22,62,102,142$ qubits)  compared with the bound $s=\lceil 1/(\epsilon^{2} \delta ) \rceil$~\cite{flammia2011direct} are independent of system size when expectation values are known exactly.
(b) Finite-shot estimation. For $\rho \in \mathcal{S}$, we show the maximal shot number $\langle N_\rho \rangle$ (dashed black) versus the optimal shot cost (points; solid black fit),  attainable if exact expectation values were known (see SM~\ref{App:MCMC}) for precision $\epsilon = 5\cdot 10^{-2}$ and failure probability $\delta = 1\cdot 10^{-2}$. Shown are the average shot number over $s = 10^4$ Paulis, $\langle N_\rho\rangle  =  N_\rho/10^4$.  samples. Both bounds scale as $\sim d_{\mathcal{S}}^2$.}
    \label{fig:MCMC}
\end{figure}

{\em Local depolarization errors are scar-subspace agnostic.} We consider the  subspace spanned by \Eq{eq:ss1} and compute \Eq{eq:scarsubspacefid} for local depolarization error. The first and second terms in~\Eq{eq:scarsubspacefid} can be written, respectively,
\begin{subequations}
\begin{align}
    \sum_{i=1}^{d^2}\mathrm{Tr}[
        K_i^\mathcal{S} K_i^{\mathcal{S}\dagger}
    ]
    &=
    \frac{1}{d}\sum_{i=1}^{d^2}
    (1-p)^{w[W_i]}
    \mathrm{Tr}[\Pi_\mathcal{S} W_i]^2
    \,, \label{eq:term1}\\
    \sum_{i=1}^{d^2}|\mathrm{Tr}[\Pi_\mathcal{S}K_i]|^2 &= \sum_{i=1}^{d^2}\Big(1-\frac{3p}{4}\Big)^{2L-w[W_i]} \left( \frac{p}{4}\right)^{w[W_i]}\nonumber\\&\qquad\qquad\qquad\qquad\times|\mathrm{Tr}[\Pi_\mathcal{S}W_i]|^2 \label{eq:term2}
\end{align}
\end{subequations}
where $w$ is the Pauli weight of $W_i$. The only operators with non-zero trace in $\mathcal{S}$ are tensor products of the stabilizer generators, $X_p X_{p+L}$, $Y_p Y_{p+L}$, and $Z_p Z_{p+L}$, all others  are traceless. From \Eqs{eq:term1}{eq:term2}, using arguments outlined in the section~\ref{supp:fidelity} of the SM, we bound
\begin{align}\label{eq:loosebound}
&\Big(1-\frac{3p}{4} \Big)^n \le  F(\mu_\mathcal{S})\nonumber\\& \le \frac{1}{1+d_s}    +\frac{d_s}{1+d_s}
   \left[
      \Big(1-\frac{3p}{4}\Big)^2
      +3\Big(\frac{p}{4}\Big)^2
   \right]^{\frac{n}{2}}  \, .
\end{align}
The lower and upper bounds bracket $F(\mu_\mathcal{H})$ and coincide with it, and with each other, asymptotically in $n$ or for small $p$.  However, away from this regime, the upper bound \Eq{eq:loosebound}
is greater than or equal to
$F(\mu_\mathcal{H})$ by sub-leading-in-$n$ contributions. Thus the fidelity of non-simulable states may fall below that of the scar benchmark for larger $p$, in which case it may not provide reliable evidence of quantum advantage. Worse, the upper bound \Eq{eq:loosebound} is loose at $p=\mathcal{O}( 1)$.  To obtain tight bounds beyond this ``weak-noise'' regime (large $n$ or small $p$), we exploit the structure of the scar subspace outlined in SM~\ref{supp:fidelity} to upper and lower bound $|\mathrm{Tr}[\Pi_\mathcal{S}W]|$ for the selected operators $W$; these have no compact form and are not written in the main text.

\Fig{fig:noisebounds} illustrates the tightness of these combined bounds, where the top panel compares the fidelity of states in the full Hilbert space $F(\mu_{\mathcal{H}})$ (solid lines) with that of states in the scar subspace $F(\mu_{\mathcal{S}})$ (shaded bands). The difference shown in the bottom panel vanishes as $p \to 0$ or large $n$, and trivially as $p\rightarrow 1$. The weak-noise regime $p \lesssim 1/n$ is contrasted with the constant-noise regime $p = \mathcal{O}(1)$, and the transition between the two bounds is marked by square markers in the bottom panel; the tighter bound is always shown.
The averaged fidelity of non-verifiable typical states is as good as that of states in the scar subspace.
Our bounds hold \textit{on average}, i.e., input states to the error channel are drawn from the respective Haar measure rather than generated through circuit evolution. Below, we  show numerically that the fidelities nearly coincide for  circuit evolution from scar and non-scar initial states.

{\em Numerical Simulations.} We present numerical results using a Markov Chain Monte Carlo (MCMC) sampling algorithm for DFE, \Eq{eq:def_fidelity}. The algorithm proposes only Pauli operators  that have nonzero matrix elements within the target scar subspace, i.e., from the set $P_{\mathcal{S}}\equiv\{ W_k \, | \, \exists \rho \in \mathcal{S} \text{ s.t. }\text{Tr}[\rho W_k]\neq 0\}$; there are $|P_{\mathcal{S}} | = (n^2-n+2 )d$ such operators for the prototype model~\Eqs{eq:Hamiltonian}{eq:ss1}; the precise algorithm is discussed in Section~\ref{App:MCMC} of the SM. 
Since our focus is on demonstrating this efficiency  rather than the MCMC sampler itself, we use a simple uniform proposal distribution over $P_{\mathcal{S}}$.

\Fig{fig:MCMC}(a) shows that the DFE estimate is $\epsilon$-close to the true fidelity with probability $1-\delta$ if $s=\lceil 1/(\epsilon^{2} \delta ) \rceil$ (dashed black line) if the Pauli expectation values are known exactly~\cite{flammia2011direct}. We estimate the fidelity between two pure target states, $|\rho\rangle$ and $|\sigma\rangle$ where $\ket{\rho}$ is chosen Haar-randomly from the scar subspace, and $\ket{\sigma}$ is obtained via time evolution within the same subspace such that $|\sigma\rangle = e^{-i H t} |\rho\rangle$, with $gt = 0.5$. For  fixed  states, we run 500 independent MCMC chains, each with a burn-in of $10^3$ iterations and a thinning interval of $100$, collecting $2\times10^4$ samples per chain. We plot the mean estimate of $\delta$ over all chains, with error bands indicating the $95\%$ Wilks confidence interval~\cite{brooks2011handbook}. As the figure shows curves for $L=11,31,51,71$ ($n=22,62,102,142$ qubits), the accuracy of the estimate is independent of $n$. 

\Fig{fig:MCMC}(b) shows the corresponding analysis for finite-shot estimation. 
The worst-case shot cost, \Eq{eq:leonebound}, (per sample) (dashed line)  is compared with the hypothetical optimal shot cost (circles with solid black fit), corresponding to the number of shots required from binomial statistics if the exact expectation values $\mathrm{Tr}[W_{k_i}\sigma]$ were known; see SM~\ref{App:MCMC} for details. Both  worst and optimal  cases exhibit the same polynomial scaling, $d_{\mathcal{S}}^2\sim n^2$. 

\begin{figure*}[t]
    \centering
    \includegraphics[width=1.0\linewidth]{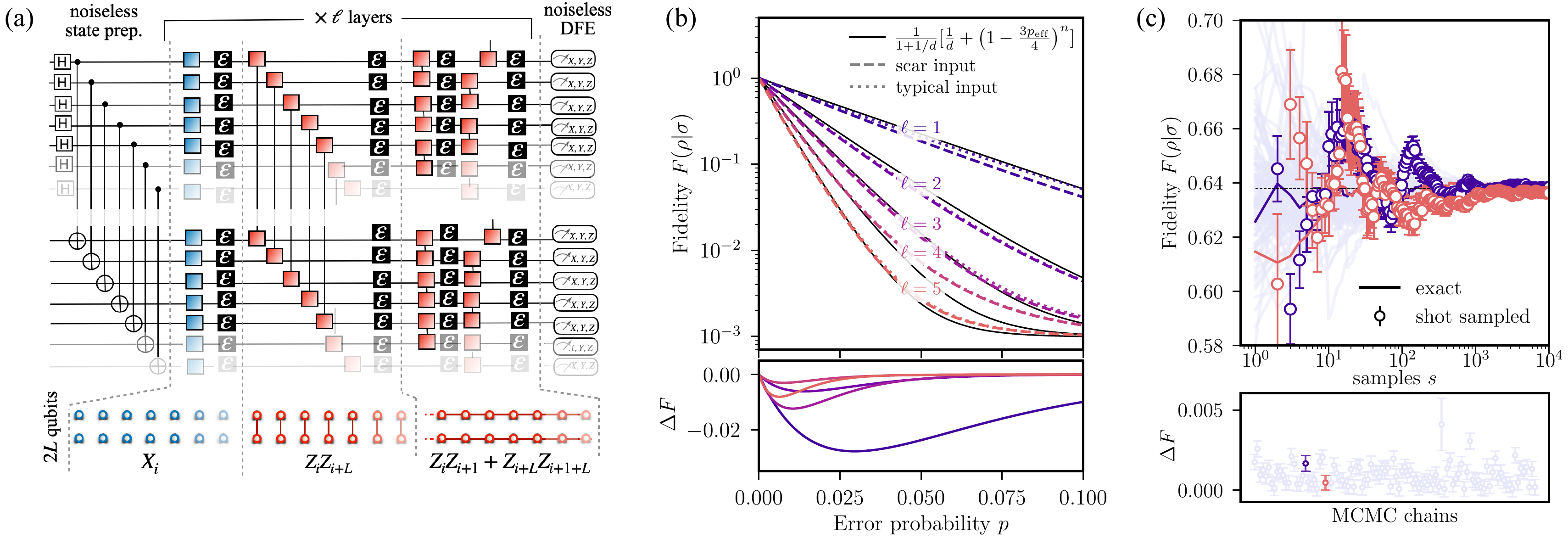}
    \caption{\textit{Circuit Simulation.} 
(a) Quantum circuit used to simulate \Eq{eq:Hamiltonian}, consisting of state preparation, (random or Trotterized) time evolution, and DFE. Local depolarizing noise is assumed, except during state preparation and measurement. The state preparation circuit consists of $L=n/2$ two-qubit Bell-pair sub-circuits and is applied to an input computational basis product state $| s \rangle$, where s is an $n$-qubit bitstring. For special $s$, the circuit prepares QMBS basis states, while for all others it produces orthogonal non-scar states.
(b) Classical emulation of the circuit in (a) for $n=10$ qubits comparing fidelities as a function of  $p$ for different circuit depths $\ell$, comparing scar basis state (dashed lines) with non-scar inputs (dotted lines). The  results are compared with the analytical expression for the {\em average} fidelity, Eq.~\eqref{eq:fidscar}, modified by an effective error rate $p_{\rm eff}=1-(1-p)^\ell$ accounting for the $\ell$ layers (solid black lines). The bottom panel shows the difference $\Delta F \equiv F_\mathcal{S} - F_\mathcal{H}$, decreasing with $\ell$. (c) Noisy circuit simulation with $n=10$ qubits, $\ell=15$ layers and $p=0.001$, where Pauli expectation values are estimated from finite-shot sampling. Top: Fidelity versus the number of MCMC samples, $s$. Circles show shot-sampled results ($10^4$ shots per Pauli observable), while solid lines correspond to exact Pauli expectation values. Two representative MCMC chains are shown out of $100$. Bottom: Final fidelity precision at the end of shot-sampled MCMC chains ($s=10^4$).}
    \label{fig:circuit}
\end{figure*}

Next, we focus on simulating noisy quantum circuits for $\sigma$. We classically emulate, using \textsc{qiskit}~\cite{javadi2024quantum}, the preparation of an initial basis state and its subsequent (noisy) time evolution, see \Fig{fig:circuit}(a) for circuit and error model. We apply single qubit depolarizing noise after every gate in the time evolution circuit but not for state preparation or measurement.  The initial state preparation circuit consists of $L=n/2$ two-qubit Bell-pair sub-circuits  applied to a product state $| s \rangle$, where s is an $n$-qubit bitstring. For special choices of $s$, the circuit prepares QMBS basis states, while for generic bitstrings it produces orthogonal non-scar states. The evolution circuit may either be Trotterized time evolution circuit or a random  circuit  preserving the subspace.

For simplicity, we employ a random time evolution circuit for $n=10$ qubits of the subspace preserving form,
\begin{align}
&U(\boldsymbol{\alpha},\boldsymbol{\beta},\boldsymbol{\gamma}) \equiv \prod_k^\ell e^{-i\sum_{i=1}^L {\alpha_{k,i}} [Z_i Z_{i+1} + Z_{i+L} Z_{i+L+1} ]} 
\nonumber\\
 &\times e^{ -i \sum_{i=1}^L {\beta_{k,i}} Z_iZ_{i+L}}  e^{- i\sum_{i=1}^L\gamma_{k,i} [X_i+X_{i+L}]}
\end{align}
where $\ell$ is the number of  evolution layers and the angles $\boldsymbol{\alpha},\boldsymbol{\beta},\boldsymbol{\gamma}$ are randomly (uniformly in $i$) chosen for every layer; single $X$ (two-qubit $ZZ$) rotations are shown in blue (red) in~\Fig{fig:circuit}(a)~\footnote{A Trotter time evolution circuit is obtained by setting $\alpha_{k,i}=\beta_{k,i}=\gamma_{k,i}=\delta t$ where $\delta t $ is the time step.}.
\Fig{fig:circuit}(b) compares the fidelity of a circuit evolved state within the scar subspace, $F(\rho_{\mathcal S},\mathcal{E}(\rho_{\mathcal S}))$, with that of a typical state outside the subspace, $F(\rho_{\mathcal H},\mathcal{E}(\rho_{\mathcal H}))$. Here, $\rho_{\mathcal S}$ is obtained by evolving the basis scar state $\varphi_{0,1}$ [Eq.~\eqref{eq:ss1}] under the ideal circuit, while $\mathcal{E}(\rho_{\mathcal S})$ is the corresponding output of the same circuit in the presence of noise. The states $\rho_{\mathcal H}$ and $\mathcal{E}(\rho_{\mathcal H})$ are defined analogously for an initial state outside the scar subspace using the same circuit, the only difference being the input string $s$. The bottom panel shows the difference which vanishes with circuit layers $\ell$.
Although limited to modest system size, the fidelities are nearly indistinguishable at $n=10$, and both are very close to average analytical results, indicating that noise does not distinguish between states in the scar subspace versus the full Hilbert space.  Along with our analytical bounds, this indicates that our protocol accurately estimates the fidelity of  classically intractable states, up to controlled corrections that vanish asymptotically. Moreover, using concentration of the Haar measure, one can show that the fidelity equivalence holds not just on average, but for most states (see SM~\ref{supp:fidelity} for details).

In \Fig{fig:circuit}(c) we emulate a noisy experiment by simulating the circuit in (a), starting from a stabilizer-scar basis state and applying local depolarizing noise with strength $p=0.001$ to a system of $n=10$ qubits evolved for $\ell=15$ layers. We generate $100$ independent MCMC chains and estimate the fidelity using either exact Pauli expectation values (solid lines) or finite-shot measurements with $10^4$ shots per Pauli observable (circles; error bars denote the propagated statistical uncertainty). The top panel shows two representative MCMC trajectories, while the bottom panel summarizes the fidelity precision achieved after $s=10^4$ MCMC samples. Consistent with \Fig{fig:MCMC}, the observed precision and failure probability substantially outperform the bound \Eq{eq:leonebound}.

{\em Discussion.} We presented a protocol for benchmarking quantum simulations at scale. The approach rests on three key ingredients: (i) classically simulability of scar states, (ii) efficient fidelity estimation via the stabilizer structure of the subspace, and (iii) that their fidelity tightly constrains that of classically non-simulable states evolved under the same circuit.
The key to this protocol is that, although the unitary dynamics differ strongly between the scar and thermalizing subspaces, local decoherence acts similarly in both.  This allows us to benchmark performance by only considering simulable dynamics, and estimating fidelity in the non-simulable regime.

There are several caveats. First, (iii) has only been proven as an average over input states to the error channel. Numerical evidence, however, strongly suggests that it also holds for individual states computed by quantum circuits. Second, our analysis assumes local depolarizing noise, and it remains to be seen how the results extend to other error models. However, while one can certainly engineer either a subspace or an error model that violates (iii), generic noise is unlikely so specific as to preserve the commutant of a nontrivial, digitally simulated Hamiltonian, let alone of arbitrary {systematically constructible  models}. Finally, it is unclear whether local stochastic noise is an effective description for logical qubits. Although logical error models can be derived  under simplifying assumptions~\cite{rahn2002exact,iyer2022efficient}, recent work suggests the possibility of non-Markovian logical noise~\cite{sutherland2018non,ziyad2025emergent};  for approximate decoding~\cite{demarti2024decoding,roffe2005decoding}  logical errors are even less understood. Many practical improvements are  also possible, for instance our MCMC algorithm is not efficient.

A next step is is the implementation of the protocol on  quantum hardware, where achieving fidelities beyond those attainable by classical algorithms  would constitute a demonstration of quantum advantage. Our goal is benchmarking discovery-scale quantum simulations of nonequilibrium dynamics and thermalization in structured models such as LGTs~\cite{mueller2025quantum,mueller2022thermalization,zhou2022thermalization}, with applications to high-energy and nuclear physics, condensed matter, and materials science.

\textit{Acknowledgments.} We thank Andreas Elben, Lukasz Fidkowski, and Cole Maurer for discussions. This material is based upon work supported by the U.S. Department of Energy, Office of Science, National Quantum Information Science Research Centers, Quantum Systems Accelerator (Award No. DE-SCL0000121). J.H. is supported, in part, by U.S. Department of Energy, Office of Science, Office of Nuclear Physics, InQubator for Quantum Simulation (IQuS) under Award Number DOE (NP) Award DE-SC0020970 via the program on Quantum Horizons: QIS Research and Innovation for Nuclear Science, and, in part, by the Department of Physics and the College of Arts and Sciences at the University of Washington.

\bibliography{bibi.bib}
\appendix
\section*{Supplemental Material}

\setcounter{subsection}{0}

\renewcommand{\thesubsection}{S\arabic{subsection}}

\renewcommand{\theequation}{S\arabic{equation}}
\setcounter{equation}{0}

\renewcommand{\thefigure}{S\arabic{figure}}
\setcounter{figure}{0}

\subsection{Fidelity Estimation in Scar-subspace}\label{App:efficient_fidelity_estimation}
In this section, we show that direct fidelity estimation (DFE) is efficient for states in the scar-subspace. We do this by bounding the $\alpha$-SRE of any state $\rho \in \mathcal{S}$ (target state). 
\subsubsection{Upper-Bound on $\alpha$-SRE of $\rho$}
We start by noting that $P_{\rho}(k_{i})=\text{Tr}[\rho W_{k_{i}}]^{2}/d$ is a probability distribution for a pure state.
Using this we define the stabilizer Renyi entropy of $\rho \in \mathcal{S}$ as \cite{leone2022stabilizer}:
\begin{equation}\label{eq:app:def_stab_renyi_entropy}
    M_{\alpha} \equiv \frac{1}{1-\alpha} \log(\xi_\alpha)\,, \quad \xi_\alpha \equiv \frac{1}{d}\sum_{W_k\in \mathcal{P}_{n}} \text{Tr}[\rho W_{k}]^{2\alpha}\,.
\end{equation}
Note that $M_0$ is defined as the logarithm of the number of Paulis with a non-zero expectation value. Consider also 
the stabilizer fidelity of the state  $\ket{\psi} $ defined as \cite{wei2003geometric}:
\begin{align}\label{eq:stabilizerfidelitydef}
    \mathcal{F} (| \psi \rangle ) \equiv \max_{\ket{\phi} \in \text{STAB}} | \langle \phi | \psi \rangle |^2,
\end{align}
where the maximization runs over the set of all stabilizer states ($\text{STAB}$). We now show that the stabilizer fidelity of the target state $\rho := \ketbra{\psi_f} \in \mathcal{S}$ is lower bounded in the following way:
\begin{equation}
    \mathcal{F}(\ket{\psi_{f}} ) = \max_{\ket{\phi} \in \text{STAB}} | \langle \phi | \psi_{f} \rangle |^2 \geq \frac{1}{d_{s}}. 
\end{equation}

This is seen as following. First, note that given an orthonormal basis $S= \{x_{1}, x_{2}, \dots, x_{d_{s}} \}$ in $\mathbb{C}^{d_{s}}$, the maximum inner product between a given normalized vector $\ket{\psi} \in \mathbb{C}^{d_{s}}$ and any of the basis vectors is bounded by:
\begin{equation}\label{eq:packing}
    c_{\max}^{2} = \max_{i} |\langle x_{i}|\psi\rangle|^{2}  \geq \frac{1}{d_{S}}\,,
\end{equation}
where $d_S$ is the dimension of $\text{span}(S)$.
This bound holds even if we extend the set $ S $ with an arbitrary element $v \in\mathbb{C}^{d_{s}} $. 
\begin{figure*}[t]
    \centering
    \includegraphics[width=0.9\linewidth]{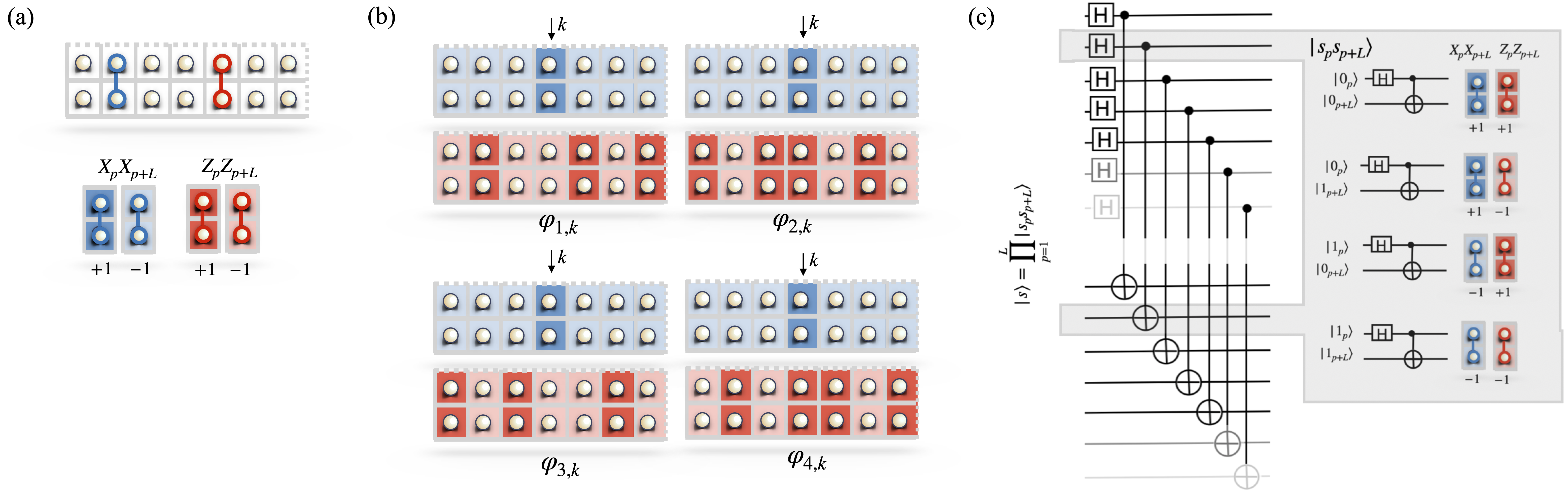}
    \caption{\textit{Stabilizer scar subspace of the prototype model.}
(a) The model, \Eq{eq:Hamiltonian} of the main text, is defined on an $L \times 2$ lattice
with periodic boundary conditions. The dual lattice of the LGT,
Eq.~\eqref{eq:LGT}, is shown in gray, in which the degrees of freedom reside
on the links.
(b) For odd $L$, the QMBS subspace is spanned by the orthonormal stabilizer
basis states shown here, each characterized by a distinct eigenvalue pattern of
the stabilizer generators $X_i X_{i+L}$ (blue) and $Z_i Z_{i+L}$ (red), where
darker (lighter) shading indicates a $+1$ ($-1$) eigenvalue. (c) Initial states are prepared using simple Bell circuits specified by an input bit string $s$, which can generate any QMBS basis state or a state outside the scar subspace with the same circuit complexity. Subfigures (a) and (b) adapted from Ref.~\cite{hartse2025stabilizer}.}
    \label{fig:model}
\end{figure*}

Take $S$ to be the stabilizer basis of the scar subspace $\mathcal{S}$. The stabilizer fidelity of $\rho$ is given by
\begin{equation}\label{eq:stab_fidelity_maximization}
    \mathcal{F} (| \psi \rangle ) \equiv \max_i |{_{i}\langle \phi |} \psi_{f} \rangle |^2,
\end{equation}
where $\ket{\phi}_{i} \in \mathrm{STAB} =\{ \ket{\phi}_{1},\ket{\phi}_{2}, \ket{\phi}_{3},\dots, \ket{\phi}_{m}  \}$, the set of all $n$-qubit stabilizer states. Since $\ket{\psi_f} \in \mathcal{S}$, the maximization over all stabilizer states in the RHS of \cref{eq:stab_fidelity_maximization} reduces to maximization over $\text{STAB} \cap\mathcal{S}$, which is an extension of the set $S$. Then, using Eq.~\eqref{eq:packing}, we get
\begin{equation}\label{eq:scarstabfidelitybound}
    \mathcal{F} (| \psi \rangle ) \equiv \max_i |_i{\langle \phi |} \psi_{f} \rangle |^2 \geq \frac{1}{d_{s}}.
\end{equation}

Now, we use the lower bound on the stabilizer fidelity to lower-bound the average value of $\text{Tr}[\rho W_{k_{i}}]^{2}$. This is done by using the connection between stabilizer fidelity with the $\alpha$-SRE (for $\alpha > 1$)\cite{haug2023stabilizer}:
\begin{align}\label{eq:relativeebtropyandstabfidbound}
    M_\alpha (| \psi \rangle ) \le \frac{2\alpha}{\alpha - 1}  \log[\mathcal{F} (| \psi \rangle )]^{-1}.
\end{align}
Combining \cref{eq:scarstabfidelitybound} with \cref{eq:relativeebtropyandstabfidbound} yields:
\begin{equation}\label{eq:boundsonxi}
\begin{split}
     M_\alpha (| \psi \rangle ) &\le  \frac{2\alpha}{\alpha - 1}  \log(d_{s} ).
\end{split}
\end{equation}
Thus, $\alpha$-SRE is logarithmically upper-bounded for all states in the scar subspace for $\alpha >1$. In \cref{eq:countingPS} we show that the $0$-SRE is upper-bounded logarithmically in system size.

\subsubsection{Connection of $\alpha$-SRE to DFE}

Recall from \cref{eq:boundDFE} of the main text that the sampling cost of DFE to guarantee an $\epsilon$-$\delta$ convergence is related to the minimum expectation value of the sampled Pauli  $W_{k_{i}}$. The upper-bound in \cref{eq:boundsonxi} implies that $\xi_{2}$, which is the average value of random variable $\Tr[\rho W_{k_{i}}]^{2}$ as seen in \cref{eq:app:def_stab_renyi_entropy}, is lower-bounded as:
\begin{equation} \label{eq:lowerboundonexpenctationvalue}
\begin{split}
    &\underset{k_{i}\sim P_{\rho}}{\mathbb{E}}[\Tr[\rho W_{k_{i}}]^{2}] \geq \frac{1}{d_{s}^{4}}. 
\end{split}
\end{equation}
Consider the worst-case scenario where Eq.~\eqref{eq:lowerboundonexpenctationvalue} is saturated.  Since $\Tr[\rho W_{k_{i}}]^{2}$ is a bounded random-variable, it is sub-gaussian by Hoeffding's lemma:
\begin{equation}\label{eq:Hoeffdings_tail_bound}
    \text{Pr}\left[\left|\Tr[\rho W_{k_{i}}]^{2} - \underset{k_{i}\sim P_{\rho}}{\mathbb{E}}[\Tr[\rho W_{k_{i}}]^{2}] \right| \geq \gamma\right]\leq 2 \exp \left ( -2\gamma^{2}\right)\,,
\end{equation}
where $\gamma>0$ is a constant.
Thus, it us exponentially unlikely for a Pauli expectation values to be significantly smaller than the average in \cref{{eq:lowerboundonexpenctationvalue}}. Combining \cref{eq:lowerboundonexpenctationvalue}, \cref{eq:Hoeffdings_tail_bound} and \cref{eq:boundDFE} immediately implies that DFE is efficient.

More rigorously, all Pauli expectation values that are smaller than \cref{{eq:lowerboundonexpenctationvalue}} can be eliminated using the `$\epsilon$-truncation trick', wherein, one pays a small bias as a cost for the truncation (as seen in \cite{flammia2011direct,leone2023nonstabilizerness}). This is done by constructing a truncated state $\rho_{2}$, where any Pauli expectation with $|\Tr{\rho W_{k}}| < \epsilon/d_{s}^2$ are ignored. Doing DFE with $\rho_{2}$ gives the correct fidelity up to an asymptotically vanishing bias: $|\Tr{\rho_{2}\sigma} - \Tr{\rho\sigma}| \leq \epsilon/O(d_{s})$. Since $\rho_{2}$ has bounded Pauli expectation values by definition, $|\Tr{\rho_{2} W_{k}}| \geq \epsilon/d_{s}^2$, the shot sampling cost of DFE is 
\begin{equation}\label{eq:bound}
\begin{split}
    \underset{k_{i}\sim P_{\rho_{2}}}{\mathbb{E}} [m_{i}] = \underset{k_{i}\sim P_{\rho_{2}}}{\mathbb{E}} \left\lceil \frac{2\log(2/\delta)}{\text{Tr}[\rho_{2} W_{k_{i}}]^{2} l \epsilon^{2}} \right\rceil \leq \left \lceil \frac{2d_{s}^4\log(2/\delta)}{ \epsilon^{4}} \right \rceil. 
\end{split}
\end{equation}
A tighter and more general version of this argument is presented in Ref.~\cite{leone2023nonstabilizerness}. Considering total number of copies of the state $
N_{\rho} ={\mathbb{E}}[m_{i}] = \sum_{i=1}^{l}m_{i}$, Leone et al. \cite{leone2023nonstabilizerness} show that
\begin{equation}\label{eq:leone_bound}
    \frac{2}{\epsilon^2}\log(2/\delta) \exp[M_{2}] \leq N_{\rho}^{\text{\cite{leone2023nonstabilizerness}}} \leq \frac{64}{\epsilon^4}\log(2/\delta) \exp[M_0],
\end{equation}
where $M_{2}$ and $M_{0}$ are the 2- and 0-SREs of $\rho$, respectively. Inserting in the upper-bound on $M_{0}$ from \cref{eq:countingPS},  directly leads to \cref{eq:leonebound} of the main text.

In the main text we show that upper-bound derived from \cref{eq:leone_bound} scales as $\mathcal{O}(d_{s}^2)$ as compared to $\mathcal{O}(d_{s}^4)$ in \cref{eq:bound}. This may be a consequence of the looseness of \cref{eq:relativeebtropyandstabfidbound}. Numerical evidence presented in Ref.~\cite{haug2023stabilizer} suggests that the prefactor appearing in Eq.~\eqref{eq:relativeebtropyandstabfidbound} may be smaller by a factor of half for $1\leq \alpha \leq 2$,  implying $\xi_2 \gtrsim d_s^{-2}$, in which case Eq.~\eqref{eq:leone_bound} and Eq.~\eqref{eq:bound} scale identically. 

\subsection{QMBS Model Details}\label{supp:model}
In this section of the Supplemental Material, we provide additional details
about the model, introduced in Ref.~\cite{hartse2025stabilizer} and in \Eq{eq:Hamiltonian} of the main text, and its QMBS subspace. The model is not unique in its hosting of stabilizer scars, and other models have been found or can be systematically constructed~\cite{dooley2026parent,hokkyo2026exact,gupta2026exact}.

Details were previously presented in Ref.~\cite{hartse2025stabilizer}, and we provide only a high level summary. The model, \Eq{eq:Hamiltonian}, is defined on a rectangular $L\times 2$ lattice with periodic boundary conditions and parity constraint $\prod_{p=1}^{2L}X_p=1$, in which case it is dual to $\mathbb{Z}_2$ LGT on $L\times 2$ plaquettes with Hamiltonian~\cite{wegner1971duality,horn1979hamiltonian,sachdev2019topological},
\begin{align}\label{eq:LGT}
    H_{\mathbb{Z}_2}=-\sum_{p} \prod_{l \in p}\sigma^x_{l} - g\sum_l \sigma^z_l\,,
\end{align}
where $l$ labels a link of the dual lattice, and $p$ labels a plaquette. The model is subject to $2L-1$ independent Gauss law constraints, $G_s\equiv \prod_{\ell \in s} \sigma^z_\ell$, with $G_s | \psi \rangle =| \psi \rangle  $ for ``allowed/physical'' states $| \psi\rangle$, where $s$ labels a site in the dual lattice, thus matching the number of d.o.f. in~\Eq{eq:Hamiltonian} of the main text.

The stabilizer QMBS structure of the model originates from an underlying
commutant algebra generated by fermionic operators. As shown in
Ref.~\cite{hartse2025stabilizer}, the spin model,~\Eq{eq:Hamiltonian} of the main text, admits
a Majorana representation via the Jordan-Wigner (JW) transformation. In this
formulation, the Majoranas can be grouped into different
fermions, and the QMBS subspace corresponds to the conserved single-particle
and single-hole subspace when $L$ is odd. In this case, the operator
$(\hat{N}-L)^2$, where $\hat{N}$ is the total fermion number operator,
generates the full commutant of the scar subspace.
In the spin model, the stabilizer generators $X_i X_{i+L}$ and $Z_i Z_{i+L}$
(for $i = 1,\dots,L$), see \Fig{fig:model}(a),  correspond to relative parity operators between adjacent
fermionic modes. The orthonormal stabilizer states spanning the QMBS subspace
are depicted in Fig.~\ref{fig:model}(a) and (b), adapted from Ref.~\cite{hartse2025stabilizer}.
The matrix elements of the Hamiltonian within the scar subspace, which grows
only linearly in the qubit number $n = 2L$, are derived in the Supplemental
Material of that reference, enabling efficient simulation of scar dynamics for
systems of arbitrary size. 

State preparation of stabilizer-scar basis states is particularly straightforward. \Fig{fig:model}(c) illustrates the circuit used to prepare any such state. The circuit consists of a simple Bell-state preparation circuit acting on an input bitstring $s$, which is generated by applying \textsc{not} ($X$) gates to the corresponding qubits. Certain input bitstrings produce basis states of the form shown in \Fig{fig:model}(b), corresponding to stabilizer-scar states, while all other bitstrings generate non-scar states. Importantly, the preparation of scar and non-scar states requires exactly the same circuit and therefore has identical circuit complexity.

When simulating noisy circuits, we employ the \textsc{qiskit} 
simulator~\cite{javadi2024quantum} with a local depolarizing noise model of single-site error
probability $p$, as illustrated in Fig.~\ref{fig:circuit}(a). The circuit is
left uncompiled (relative to QISKIT's native gateset), so that the noise channel is inserted immediately after each
gate layer as it appears in the figure. The elementary gate primitives are
single-qubit rotations $\exp\{i \alpha_i X_i\}$ and two-qubit entangling gates
$\exp\{i \alpha_{ij} Z_i Z_j\}$. The local depolarizing channel is applied after each
such gate; no noise acts on the initial-state preparation circuit shown in
Fig.~\ref{fig:model}(c), nor on the final measurement layer.

\subsection{Average fidelity of the noisy channel}\label{supp:fidelity}
In this section of the Supplemental Material, we analyze the fidelity of states within the scar subspace, and compare it with those in the full Hilbert space, under local depolarizing noise. For simplicity, we focus on the average channel fidelity with respect to input states drawn from a given measure $\mu$, as defined in Eq.~\eqref{eq:avfidel} of the main text. In particular, $\mu_{\mathcal{H}}$ denotes the Haar measure on the full Hilbert space, while $\mu_{\mathcal{S}}$ denotes the Haar measure restricted to the (polynomially sized) stabilizer scar subspace.

For the Haar measure $\mu_{\mathcal{H}}$, this can be calculated exactly following Ref.~\cite{Mele2024introductiontohaar}:
\begin{equation}
    \label{eq:averagefidelityfull}
    F(\mu_{\mathcal{H}}) = \frac{1}{d(d+1)} \Big( d+ \sum_{i=1}^{d^2}|\text{Tr}\left(K_{i} \right)|^{2} \Big),
\end{equation}
where $d $ is dimension of the Hilbert space, and $K_i$ are the Kraus operators for the local depolarizing channel with probability $p$, i.e., $K_i \in \{{(1-\frac{3p}{4})^{\frac{1}{2}}} \mathbb{I}, {(\frac{p}{4})^{\frac{1}{2}}}X, {(\frac{p}{4})^{\frac{1}{2}}}Y, {(\frac{p}{4})^{\frac{1}{2}}}Z \}^{\otimes n}$, where $X,Y,Z$ are Pauli matrices. All $K_i$ are traceless, except, $K_1=(1-\frac{3p}{4})^{n/2} \mathbb{I}$, therefore 
\begin{align}
    \sum_{i=1}^{d^2}|\text{Tr}\left(K_{i} \right)|^{2} = d^2 \Big( 1 - \frac{3p}{4}\Big)^n\,,
\end{align}
and one obtains \Eq{eq:fidscar} of the main text. Similarly, we can consider the average channel fidelity over states in the scar subspace, under the same error channel,
\begin{align}
    F(\mu_\mathcal{S})&=\sum_{i=1}^{d^2}\mathbb{E}_{|\psi \rangle \sim \mu_{\mathcal{S}}}\Big[ \text{Tr}\Big( | \psi \rangle \langle \psi| K_i| \psi \rangle \langle \psi| K_i^\dagger  \Big)\Big]\nonumber
    \\ &=\sum_{i=1}^{d^2}\mathbb{E}_{|\psi \rangle \sim \mu_{\mathcal{S}}}
    \Big[
\text{Tr}\Big( 
| \psi \rangle \langle \psi|^{\otimes 2} (K_i \otimes K_i^\dagger) \mathbb{F}
\Big)
    \Big]\,,
\end{align}
where  $\mathbb{F}$ is the swap operator. Since  $|\psi \rangle \langle \psi |=\Pi_\mathcal{S}|\psi \rangle \langle \psi | \Pi_\mathcal{S} $ for states in $\mathcal{S}$, projectors $\Pi_\mathcal{S}$ on the subspace can be introduced,
\begin{align}\label{eq:longuglyscar}
&\sum_{i=1}^{d^2}\mathbb{E}_{|\psi \rangle \sim \mu_{\mathcal{S}}}
    \Big[
\text{Tr}\Big( 
| \psi \rangle \langle \psi|^{\otimes 2} (\Pi_\mathcal{S}K_i\Pi_\mathcal{S} \otimes \Pi_\mathcal{S}K_i^\dagger\Pi_\mathcal{S}) \mathbb{F}
\Big)
    \Big]\nonumber\\
    =& \sum_{i=1}^{d^2}\text{Tr}\Big(\frac{\mathbb{I}_\mathcal{S}+\mathbb{F}_\mathcal{S}}{d_s(d_s+1)}(\Pi_\mathcal{S}K_i\Pi_\mathcal{S} \otimes \Pi_\mathcal{S}K_i^\dagger\Pi_\mathcal{S}) \mathbb{F}\Big)
    \nonumber\\
    =& \sum_{i=1}^{d^2} \frac{1}{d_s(d_s+1)}\Big\{  \text{Tr}(K_i^\mathcal{S} \otimes K_i^{\mathcal{S}\dagger}  \mathbb{F} + K_i^\mathcal{S} \otimes K_i^{\mathcal{S}\dagger} \mathbb{F} \mathbb{F}_\mathcal{S})\Big\}
    \nonumber\\
    = &\sum_{i=1}^{d^2} \frac{1}{d_s(d_s+1)}\Big\{ \text{Tr}(K_i^\mathcal{S}K_i^{\mathcal{S}\dagger})+| \text{Tr}(K_i^\mathcal{S})|^2\Big\}\,,
\end{align}
yielding \Eq{eq:scarsubspacefid} of the main text. We have abbreviated $K_i^\mathcal{S}\equiv \Pi_\mathcal{S}K_i\Pi_\mathcal{S}$ and $\mathbb{F}_\mathcal{S}$  and $\mathbb{I}_\mathcal{S}$ are the swap and identity operators in $\mathcal{S}$, respectively; we used that $\mathbb{F}$ commutes with operators of the form $\hat{O}\otimes \hat{O}$, and that $\mathbb{F}\mathbb{F}_\mathcal{S}$ acts like identity on states within the scar subspace, i.e., $(\Pi_\mathcal{S} \otimes \Pi_\mathcal{S})\mathbb{F}\mathbb{F}_\mathcal{S}(\Pi_\mathcal{S} \otimes \Pi_\mathcal{S})=\mathbb{I}_\mathcal{S}$.

We  will now proceed to bound both terms. We start by noting that when $p = 0$, $F(\mu_{\mathcal{H}}) = F(\mu_{\mathcal{S}}) = 1$, as in the channel $\mathcal{E}$ becomes the identity operation. Similarly, when $p = 1$, the channel maps all states to the maximally mixed state, $\mathcal{E}(\rho) = \mathbb{I}/d$. Since $(1/d)\bra{\psi} \mathbb{I} \ket{\psi} = 1/d$ for all pure states $\ket{\psi}$, it follows that $F(\mu_{\mathcal{H}}) = F(\mu_{\mathcal{S}}) = 1/d$. Hence, we want bounds on $F(\mu_{\mathcal{S}})$ that converge to $F(\mu_{\mathcal{H}})$ at $p = 0$ and $p = 1$. To achieve this, we construct two bounds. One is tight at $p = 0$ and the other is tight at $p = 1$. The final bound will be the intersection of these two bounds.

\subsubsection{Bound tight near $p=0$}
Consider the first term in \Eq{eq:longuglyscar},
\begin{align}\label{eq:firstterm}
    \sum_{j=1}^{d^2}\text{Tr}(K_i^\mathcal{S}K_i^{\mathcal{S}\dagger})=\text{Tr}\left( \Pi_\mathcal{S} \sum_{i=1}^{d^2}K_i \Pi_\mathcal{S} K_i^\dagger\right)
    =\text{Tr}(\Pi_\mathcal{S} \mathcal{E}(\Pi_\mathcal{S}))
\end{align}
where $\mathcal{E}(\cdot)$ is the error channel. Using the Pauli decomposition of $\mathcal{E}(\Pi_\mathcal{S})$, and assuming local depolarizing noise, we find
\begin{align}\label{eq:firstterm2}
     \sum_{i=1}^{d^2}\text{Tr}(K_i^\mathcal{S}K_i^{\mathcal{S}\dagger})=\frac{1}{d}\sum_{j=1}^{d^2}(1-p)^{w[W_j]} \text{Tr}(\Pi_\mathcal{S}W_j)^2
\end{align}
where $w[W_j]$ is the weight of the Pauli string $W_j$. Without considering the structure of $\text{Tr}(\Pi_\mathcal{S}W_j)$, an upper bound for this term, loose but accurate at small $p$, is
\begin{align}
\label{Eq:Weak_FTUB}
    \sum_{i=1}^{d^2}\text{Tr}(K_i^\mathcal{S}K_i^{\mathcal{S}\dagger}) \le \frac{1}{d} \sum_{j=1}^{d^2} \text{Tr}(\Pi_\mathcal{S}W_j)^2= \text{Tr}(\Pi_\mathcal{S}^2)=d_s\,.
\end{align}
 The lower bound can be found by keeping only the $K_1$ Krauss operator in  \Eq{eq:firstterm}, i.e.,
\begin{align}
\label{Eq:Weak_FTLB}
    \sum_{i=1}^{d^2}\text{Tr}(K_i^\mathcal{S}K_i^{\mathcal{S}\dagger}) &\ge \text{Tr}(K_1^\mathcal{S}K_1^{\mathcal{S}\dagger}) \nonumber \\
    &=\text{Tr}(\Pi_\mathcal{S})\Big( 1-\frac{3p}{4} \Big)^n = d_s\Big( 1-\frac{3p}{4} \Big)^n\,.
\end{align}
Notice that the above two bounds are exact at $p = 0$, and hence tight. The second term in \Eq{eq:scarsubspacefid} of the main text can be written
\begin{align}
    \label{eq:second_scar_term}
    \sum_{i=1}^{d^2}|\text{Tr}(K_i^\mathcal{S})|^2 = &\sum_{j=1}^{d^2}\Big( 1-\frac{3p}{4}\Big)^{2L-w[W_j]}\nonumber\\ &\times \Big(\frac{p}{4}\Big)^{w[W_j]}|\text{Tr}(W_j\Pi_\mathcal{S})|^2\,.
\end{align}
A straightforward upper bound can be obtained by noting that the only operators with non-zero trace in $\mathcal{S}$ are tensor products of $X_p X_{p+L}$, $Y_p Y_{p+L}$, and $Z_p Z_{p+L}$. For these, one can upper bound $|\text{Tr}(W_j\Pi_\mathcal{S})|^2\le d_s^2$ (a severe overestimate) to find
\begin{align}
    \label{Eq:Weak_STUB}
    &\sum_{i=1}^{d^2}|\text{Tr}(K_i^\mathcal{S})|^2 \le d_s^2 \sum_{j=1}^{d^2}\Big( 1-\frac{3p}{4}\Big)^{2L-w[W_j]}\Big(\frac{p}{4}\Big)^{w[W_j]}\nonumber\\
    &=\sum_{w=0,2,\dots}^{2L}\Big( 1-\frac{3p}{4}\Big)^{2L-w}\Big(\frac{p}{4}\Big)^{w} \binom{L}{w/2}3^{w/2} 
    \nonumber\\
    &=d_s^2 \left[\sqrt{\Big(1-\frac{3p}{4}\Big)^2 +3\Big(\frac{p}{4}\Big)^2}\right]^n\,.
\end{align}
A lower bound is obtained by retaining only the $i=1$ term in
\Eq{eq:scarsubspacefid} of the main text,
\begin{align}
    \label{Eq:Weak_STLB}
    \sum_{i=1}^{d^2}|\text{Tr}(K_i^\mathcal{S})|^2\geq
    |\text{Tr}(\Pi_\mathcal{S})|^2 \left(1-\frac{3p}{4}\right)^n
    = d_s^2\left(1-\frac{3p}{4}\right)^n,
\end{align}
which yields the bounds
\begin{align}
&\Big( 1-\frac{3p}{4}\Big)^n  \le F(\mu_\mathcal{S})\nonumber\\ &\le  \frac{1}{1+d_s}    +\frac{d_s}{1+d_s}
   \left[\sqrt{\Big(1-\frac{3p}{4}\Big)^2 +3\Big(\frac{p}{4}\Big)^2}\right]^n 
\end{align}
For small $p$, expanding the square root gives
\begin{align}
    \sqrt{\Big(1-\frac{3p}{4}\Big)^2 +3\Big(\frac{p}{4}\Big)^2}\approx 1-\frac{3p}{4} + \frac{3p^2}{32} +\mathcal{O}(p^3)\,,
\end{align}
so that in the large-$n$ limit the upper bound becomes
\begin{align}
&\frac{1}{1+d_s} + \frac{d_s}{1+d_s}\left(1-\frac{3p}{4}\right)^n
+ \mathcal{O}(p^2)
\nonumber\\
=& \left(1-\frac{3p}{4}\right)^n + \mathcal{O}(d_s^{-1}),
\end{align}
and the two bounds coincide asymptotically. To leading order in $p$ and
for large $n$,
\begin{align}
    F(\mu_\mathcal{S}) \approx \left(1-\frac{3p}{4}\right)^n
    \approx F(\mu_\mathcal{H}),
\end{align}
where the second approximation holds because the upper bound on
$F(\mu_\mathcal{S})$ 
{exceeds} $F(\mu_\mathcal{H})$ only by subleading
corrections.
Thus, at least asymptotically in $n$ and for small $p$, the fidelity of a
typical state is approximately equal to that of a scar subspace state, the
latter being efficiently measurable. Away from this regime, however, 
particularly for larger $p$, the bound becomes less useful since the
upper bound on $F(\mu_\mathcal{S})$ \textit{exceeds} $F(\mu_\mathcal{H})$. Hence, we now construct bounds that are tight away from the small $p$ regime. 

\subsubsection{Bound tight near $p = 1$}
To extend our result beyond the large $n$ and small $p$ regime, we proceed by finding tighter lower and upper bounds for $|\text{Tr}(\Pi_\mathcal{S}W)|$, where $W$ is a Pauli matrix, near $p=1$ (more precisely when $p  \sim \mathcal{O}(1)$). First, note that this trace is non-vanishing only if $W$ is a tensor product of the stabilizer generators defining the QMBS subspace, i.e., if $W\in \{ \mathbb{I},X_{i}X_{i+L},Y_{i}Y_{i+L},Z_{i}Z_{i+L}\}^{\otimes L}$ and for these the trace can be estimated. Specifically, we use the representation
\begin{align}\label{eq:repXXZZ}
    W=\prod_{i \in S_x \subset [L]} X_i X_{i+L}\prod_{j \in S_z \subset [L]}Z_j Z_{j+L}
\end{align}
where $[L]:= \{ 1,\dots,L\}$ denotes the set of indices of sites in the lattice (see \cref{fig:model}). The scar subspace projector can be written in terms of the basis states in \Eq{eq:ss1} of the main text, 
\begin{align}
    \Pi_\mathcal{S}=\sum_{\alpha=1}^4\sum_{k=1}^{L} | \varphi_{\alpha,k}\rangle \langle \varphi_{\alpha,k} |\,,
\end{align}
so that 
\begin{align}\label{eq:trdetail}
    \text{Tr}(\Pi_\mathcal{S}W) = \sum_{\alpha,k} \prod_{i\in S_x}\prod_{j\in S_z} \langle\varphi_{\alpha,k}| X_i X_{i+L}Z_j Z_{j+L} |\varphi_{\alpha,k}\rangle\,,
\end{align}
where we have used that if $| \psi \rangle$ is an eigenstate of operators $A,B$, then $\langle \psi | AB | \psi \rangle = \langle \psi | A | \psi \rangle \langle \psi |B| \psi \rangle $. To compute \Eq{eq:trdetail}, the key step is to rewrite the sum over all basis states in the subspace as a sum over Pauli operators. In particular, note that the expectation value of an operator labeled by $1 \le k \le L$ can be mapped to the expectation value of a corresponding operator with respect to a reference state labeled $k = 1$, e.g.,
\begin{align}
  &\mel{\varphi_{\alpha,k}}{X_i X_{i+L}\, Z_j Z_{j+L}}{\varphi_{\alpha,k}} \notag \\
  &\;= \mel{\varphi_{\alpha,1}}{X_{i-k+1} X_{i+L-k+1}\, Z_{j-k+1} Z_{j+L-k+1}}{\varphi_{\alpha,1}}\,.
\end{align}
Using the above equation and defining the cyclic permutation (in $k$) of an operator,
\begin{align}
    \label{eq:P_W}
    \mathbb{P}(W) \equiv \sum_{k'=0}^{L-1}\prod_{i\in S_x} X_{i-k'}X_{i+L-k'}\prod_{j\in S_z}Z_{j-k'}Z_{j+L-k'}\,,
\end{align}
which allows to write \Eq{eq:trdetail} as
\begin{align}
    \text{Tr}(\Pi_\mathcal{S}W) =  \sum_{\alpha=1}^4 \langle \varphi_{\alpha}| \mathbb{P}(W) | \varphi_{\alpha}\rangle\,,
\end{align}
where, from now on, we abbreviate $| \varphi_\alpha \rangle\equiv | \varphi_{\alpha,1}\rangle$
Note that the following relations hold,
\begin{subequations}
\begin{align}
|\varphi_{2} \rangle &= X_1 X_2 \dots X_{L} |\varphi_{1} \rangle \\
|\varphi_{3} \rangle &= \;\quad X_2 \dots X_{L} |\varphi_{1} \rangle\\
|\varphi_{4} \rangle &= X_1 \qquad\qquad\; |\varphi_{1} \rangle
\end{align}
\end{subequations}
which allow us to write
\begin{align}\label{eq:traceraw}
    \text{Tr}(\Pi_\mathcal{S}W) = \langle \varphi_1 | \big[ X_1 \dots X_{L}\mathbb{P}(W)X_1 \dots X_{L}  \nonumber\\+ X_2 \dots X_{L}\mathbb{P}(W)X_2 \dots X_L\nonumber\\
    +  X_1 \mathbb{P}(W)X_1+ \mathbb{P}(W) \big]| \varphi_1 \rangle\,.
\end{align}
From this expression, we now show that the traces for different operators $W$ can be readily estimated by distinguishing three cases: (i) $W$ contains no $Z_i Z_{i+L}$ terms, (ii) $W$ contains an odd number of $Z_i Z_{i+L}$ terms, and (iii) $W$ contains an even number of $Z_i Z_{i+L}$ terms.

\bigskip
\noindent{(i)} \textbf{$W$ contains no $Z_i Z_{i+L}$ terms.} In this case,
    \begin{align}\label{eq:XXonly}
        W=\prod_{i\in S_x} X_i X_{i+L}
    \end{align}
    which implies that 
    \begin{align}
        &X_1 \dots X_{L}\mathbb{P}(W)X_1 \dots X_{L}  + X_2 \dots X_{L}\mathbb{P}(W)X_2 \dots X_L\nonumber\\
    +  &X_1 \mathbb{P}(W)X_1+ \mathbb{P}(W) = 4 \mathbb{P}(W) \,,
    \end{align}
    and thus 
    \begin{align}
        \text{Tr}(\Pi_\mathcal{S}W) =4 \langle \varphi_1 | \mathbb{P}(W)  | \varphi_1\rangle\,.
    \end{align}
    To compute this note that $\langle \varphi_{1}| X_{i}X_{i+L} | \varphi_1\rangle =-1$ if $i\neq 1$ and $1$ otherwise. Thus for $W$ in \Eq{eq:XXonly}, 
    \begin{align}
        \langle \varphi_{1}| W| \varphi_1\rangle = \begin{cases}
            (-1)^{|S_x|} & \text{if }i\neq 1\,,\\
            -(-1)^{|S_x|}& \text{if } i=1\, ,
        \end{cases}
    \end{align}
    where $|S_x|$ is the cardinality of the set $S_x$. Since $\mathbb{P}(W) $ contains all permutations in $k$ of $W$, it is a sum of $L$ terms of which $|S_x|$ contain $X_1 X_{1+L}$ and $L-|S_x|$ do not.  Consequently,
    \begin{align}
        \text{Tr}(\Pi_\mathcal{S}W) &=4(L-|S_x|)(-1)^{|S_x|} - 4|S_x|(-1)^{|S_x|}    \nonumber\\
&=4(-1)^{|S_x|}(L-2|S_x|)\,,
\end{align}
    for operators of the form \Eq{eq:XXonly}.

\bigskip
\noindent{(ii)} \textbf{$W$  contains an odd number of $Z_i Z_{i+L}$ terms.}  In this case $W$ is of the form \Eq{eq:repXXZZ}, but $|S_z|=2m+1$ is odd. The terms in the cyclic permutation, $\mathbb{P}(W)$, now contain $Z_i Z_{i+L}$, and they can be grouped into two kinds: We label with $(1,\bar{1})$ any $W$ which contain $Z_1 Z_{1+L} $ as well as an even number of $Z_i Z_{i+L}$ for $i \in \{ 2, \dots, L\}$. We label with $(\bar{1},\bar{1})$ terms that contain no $Z_1 Z_{1+L} $ but an odd number of $Z_i Z_{i+L}$ for $i \in \{ 2, \dots, L\}$.

Using the fact that $X$ and $Z$ anti-commute, no matter the $X_i X_{i+L}$ content of $W$, we can tabulate the sign of the terms in Eq. \eqref{eq:P_W} (to calculate Eq. \eqref{eq:traceraw}) for the two types below
\begin{center}
\begin{tabular}{c|c|c}
  Operator & Type $(1,\bar{1})$ & Type $(\bar{1},\bar{1})$ \\
  \hline
  $\mathbb{P}(W)$ & $+$ & $+$ \\
  $X_1 \mathbb{P}(W) X_1$ & $-$ & $+$ \\
  $X_{1\cdots L}\,\mathbb{P}(W)\,X_{1\cdots L}$ & $-$ & $-$ \\
  $X_{2\cdots L}\,\mathbb{P}(W)\,X_{2\cdots L}$ & $+$ & $-$ \\
  \hline
  & $= 0$ & $= 0$ \\
\end{tabular}
\end{center}
indicating that, for given $W$, these terms cancel in \Eq{eq:traceraw} and can therefore be neglected.

\bigskip
\noindent{(iii)} \textbf{$W$  contains an even number of $Z_i Z_{i+L}$ terms.} In this case, a similar table can be constructed,
\begin{center}
\begin{tabular}{c|c|c}
  Operator & Type $(1,\bar{1})$ & Type $(\bar{1},\bar{1})$ \\
  \hline
  $\mathbb{P}(W)$ & $+$ & $+$ \\
  $X_1 \mathbb{P}(W) X_1$ & $-$ & $+$ \\
  $X_{1\cdots L}\,\mathbb{P}(W)\,X_{1\cdots L}$ & $+$ & $+$ \\
  $X_{2\cdots L}\,\mathbb{P}(W)\,X_{2\cdots L}$ & $-$ & $+$ \\
\end{tabular}
\end{center}
but now only terms of type $(1,\bar{1})$ are canceled by the symmetrization.
Since Pauli expectation values are bounded in magnitude by $1$, for operators
containing an even number of $Z_i Z_{i+L}$ terms, the trace satisfies
\begin{equation}
  \Tr(\Pi_\mathcal{S} W) \leq 4(L - |S_z|).
\end{equation}
To obtain a lower bound, we note that since $L$ is always odd in our model, the
total number of surviving Pauli terms is odd. As each such term contributes
$\pm 1$, the sum cannot vanish, and therefore
\begin{equation}
  1 \leq \bigl(\Tr(\Pi_\mathcal{S} W)\bigr)^2.
\end{equation}
Indeed, in the extreme case all but one term cancels, leaving a net contribution
of $\pm 1$. 

Thus in summary, we get the exact expression
\begin{align}
    \text{Tr}(\Pi_\mathcal{S}W) =\begin{cases}
        4(L-2|S_x|)(-1)^{|S_x|} & \text{if }|S_z|=0\\
        0 & \text{if } |S_z| \text{ odd}\,,
    \end{cases}
\end{align}
and, if $|S_z|\neq 0$ and $|S_z|$ even, we obtain the following bounds, irrespective of $S_x$,
\begin{align}
    1 \le  \text{Tr}(\Pi_\mathcal{S}W)^2 \le 16(L-|S_z|)^2\,.
\end{align}
Thus \Eq{eq:firstterm2} can be expressed,
\begin{align}
     & \sum_{i=1}^{d^2} \text{Tr}(K_i^s K_i^{\dagger s})
  = \frac{1}{d}\sum_{i=1}^ {d^2}(1-p)^{w(W_i)}\big[\text{Tr}(\Pi_\mathcal{S} W_i)\big]^2 \nonumber\\
  &=  \frac{1}{d}\sum_{S_x \subseteq [L]}\sum_{S_z \subseteq [L]}
    (1-p)^{w[W_{S_x,S_z}]}
    \Bigl[\text{Tr}\big(\Pi_\mathcal{S} W_{S_x,S_z}\big)\Bigr]^2,\nonumber \\
    &=\frac{1}{d}\sum_{S_x \subseteq [L]}\sum_{\substack{S_z \subseteq [L]\nonumber\\ |S_z|=\text{even}}}
    (1-p)^{2(|S_x|+|S_z|)}
    \Bigl[\text{Tr}\big(\Pi_\mathcal{S} W_{S_x,S_z}\big)\Big]^2, \nonumber\\
    &=\frac{1}{d}\sum_{\substack{S_x \subseteq [L]\\ S_z = \emptyset}}
      (1-p)^{2|S_x|}
      \Bigl[4(L - 2|S_x|)\Bigr]^2
    \nonumber\\& +\frac{1}{d} \sum_{\substack{S_x, S_z \subseteq [L]\\  |S_z|\neq 0,  \\ |S_z|=\text{even}}}
      (1-p)^{2|S_x|+2|S_z|}
      \text{Tr}(W_{S_x S_z}\,\Pi_\mathcal{S})^2
\end{align}
where we have now written $ W_{S_x,S_z}$ to make the dependence clear. There are $\binom{L}{|S_x|}$ terms that contain $X_i X_{i+L}$ $|S_x|$ times, and similarly for $S_z$. Hence, we arrive at the following upper bound
\begin{align}\label{eq:firsttermupper}
    \sum_{i=1}^{d^2} \text{Tr}(K_i^s K_i^{\dagger s}) \le \frac{16}{d} \Bigg[\sum_{|S_x|=0}^{L}(1-p)^{2|S_x|}\binom{L}{|S_x|}(L-2|S_x|)^2
 \nonumber\\
  +\big[1+(1-p)^2\big]^L
    \sum_{|S_z|=2, \text{even}}^{L-1} (1-p)^{2 |S_z|}\binom{L}{|S_z|}(L-|S_z|)^2 \Bigg]
\end{align}
where we used  the binomial identity 
\begin{align}
    \sum_{|S_x|=0}^L(1-p)^{2|S_x|}\binom{L}{|S_x|}=[1+(1-p)^2]^L\,.
\end{align}
\Eq{eq:firsttermupper} can be evaluated using \textit{Mathematica} yielding the closed form

\begin{align}
     \label{Eq:FTUB}
     \sum_{i=1}^{d^2} \text{Tr}(K_i^s K_i^{\dagger s})  &\leq  \frac{8}{d} L t^{2\!L-2}  \Bigg[\! (L\!+t\!-\!1\!) + (L\!-\!t\!+1)y^{L\!-\!2}  \nonumber \\
     &~~+ \frac{2 (L-1)(y^2 -1)}{t^{L -2}} \Bigg],
\end{align}
 where $t = 1 + (1 - p)^2$ and $y = (2 - t)/t$.
Similarly, the lower bound can be evaluated
\begin{align}
    \label{Eq:FTLB}
    \sum_{i=1}^{d^2} \text{Tr}(K_i^s K_i^{\dagger s}) &\ge \frac{8}{d} t^{2L} \Bigg[1+y^L +\frac{2 (L - 1) \left(1 + Ly^2\right)}{t^L}\Bigg],
\end{align}
Because we only bound the traces of terms where $|S_z|\neq 0$ and $|S_z|$ even, the bounds derived above in \Eqs{Eq:FTUB}{Eq:FTLB} are looser than the ones derived in the previous subsection in \Eqs{Eq:Weak_FTUB}{Eq:Weak_FTLB} near $p = 0$. However, because we explicitly track the weights $(1 - p)^{w[W_i]}$,  the bounds in \Eqs{Eq:FTUB}{Eq:FTLB} become tighter for larger values of $p$. Since both sets of bounds hold for all $0 \le p \le 1$, we obtain uniformly tighter bounds by taking their intersection.  

The term in Eq.\eqref{eq:second_scar_term} can be expressed as, 
\begin{align}
    &\sum_{i=1}^{d^2}|\text{Tr}(K_i^\mathcal{S})|^2\! =\! \sum_{i=1}^{d^2}\!\Big( 1-\frac{3p}{4}\Big)^{2L-w[W_i]}\Big(\frac{p}{4}\Big)^{w[W_i]}\!\text{Tr}(W_i\Pi_\mathcal{S})|^2, \nonumber \\
    &=\!\sum_{\substack{Sx \subseteq [L] \\ S_z = \emptyset}}\! \left(1 - \frac{3p}{4}\right)^{2(L -|Sx|)} \!\left(\frac{p}{4}\right)^{\!2|Sx|} \left| \operatorname{Tr}\left( \Pi_\mathcal{S}W_{S_x \emptyset} \right) \right|^2 \nonumber \\
    &+\!\sum_{S_x \subseteq [L]} \!\sum_{\substack{Sz \subseteq [L] \\ |S_z| = \text{even} \\ |S_z| \neq 0 }} \! \left(1 - \frac{3p}{4}\right)^{2(L -|Sx| - |Sz|)}\! \left(\frac{p}{4}\right)^{2|Sx| + 2|Sz|} \nonumber\\&\qquad\qquad\qquad\qquad\times \left| \operatorname{Tr}\left( \Pi_\mathcal{S}W_{S_x S_z} \right) \right|^2.
\end{align}
Using the same lines of reasoning as above, we arrive at the following upperbound, 
\begin{align}
    &\sum_{i=1}^{d^2}|\text{Tr}(K_i^\mathcal{S})|^2 \le 16 \sum_{|S_x| = 0}^L \left(1 - \frac{3p}{4}\right)^{2(L -|Sx|)} \left(\frac{p}{4}\right)^{2|Sx|}\nonumber\\&\times \left(L - 2 |S_x|\right)^2 
    +16 \left[\left(1 - \frac{3p}{4}\right)^2 + \left(\frac{p}{4}\right)^2 \right]^L \nonumber \\ 
    & \times \!\left( \!\sum_{|S_z| = 2, \text{even}}^{L}\!{\binom{L}{|S_z|}}\! \left(1 - \frac{3p}{4}\right)^{-2|S_z|} \!\left(\frac{p}{4}\right)^{2|S_z|}\! \left(L\! -\! |S_z| \right)^2\right), 
\end{align}
which when evaluated using \textit{Mathematica} gives, 
\begin{align}
   \label{Eq:STUB}
    \sum_{i=1}^{d^2}|\text{Tr}(K_i^\mathcal{S})|^2 &\le 8 L r^{L}  \Bigg[ \frac{8(1 - L)(1 - w)}{w^2} +(L+w - 1)w^{L-2}  \nonumber \\
    &~~~~~~~~~~~+ (L+ 1 - w)(2 - w)^{L - 2} \Bigg],
\end{align}
where,  $r = 1 - \tfrac{3p}{2} + \tfrac{5p^2}{8}$ and $w = 1 + (\tfrac{p}{4 - 3p})^2$. Similarly the lower bound can be evaluated as, 
\begin{align}
    \label{Eq:STLB}
    \sum_{i=1}^{d^2}|\text{Tr}(K_i^\mathcal{S})|^2 &\ge 8 L r^{L} \Bigg[ w^L\!\! + \!(2 - w)^L\! +\! \frac{8L(L-1)(1 - w)}{w^2}\Bigg].
\end{align}
It is easy to verify that the upper and lower bounds on the second term derived in Eqs.\eqref{Eq:STUB} and \eqref{Eq:STLB} are tighter than the ones derived in Eqs.\eqref{Eq:Weak_STUB} and \eqref{Eq:Weak_STLB}, for all $0 \le p \le 1$, respectively. Thus, putting everything together, the final bounds are
\begin{align}
    &\frac{1}{d_s(d_s + 1)}  \left[\max (\eqref{Eq:Weak_FTLB}, \eqref{Eq:FTLB} ) + \eqref{Eq:STLB}\right] \le  F(\mu_\mathcal{S}) \nonumber \\
    &\leq \frac{1}{d_s(d_s + 1)} \left[\min (\eqref{Eq:Weak_FTUB} , \eqref{Eq:FTUB} ) + \eqref{Eq:STUB}\right]\,.
\end{align}

\subsubsection{Concentration of Haar fidelity}
In the main text, we analyze the average fidelity. Here we provide additional arguments that this estimate holds not only on average but for the overwhelming majority of states. Whether this alone is sufficient for quantum simulation is a separate question, as quantum-simulated states may constitute a highly structured subset of Hilbert space. For the present application to thermalization, however, this assumption is well motivated. In any case, our numerical results provide strong evidence that the fidelity equivalence also holds for specific quantum-simulated states.

Levy's lemma \cite{ledoux2001concentration}, implies that the probability of a state having channel a  fidelity that is far from the average channel fidelity is exponentially small. Specifically, for a system with $n$ qubits, 
 \begin{align}
      &\underset{\substack{\ket{\phi}\sim\mu_H}}{\mathrm{Pr}} \big[ \big| \bra{\phi} \mathcal{E}(\ketbra{\phi}{\phi}) \ket{\phi} - \underset{{ \underset{\sim\mu_H}{\ket{\psi}}}}{\EE}  \left[  \bra{\psi} \mathcal{E}(\ketbra{\psi}{\psi}) \ket{\psi} \right]   \big|  \geq \varepsilon  \big]  \nonumber \\
      &\qquad \qquad \leq 2\exp \left(- \frac{ 2^n \varepsilon^2}{72\pi^3}  \right)
 \end{align}
The use of Levy's lemma requires that the channel fidelity is a Lipschitz continuous function. While this is a well-known fact, we provide an elementary proof for the sake of completeness.  
\begin{lemma}
   \label{Lem:Lip_cont}
    The function $f:\CC^{d} \to \RR $ 
    \begin{equation}
        f(\ket{\psi}) := \bra{\psi} \mathcal{E}(\ketbra{\psi}{\psi}) \ket{\psi},
    \end{equation}
    is Lipschitz continuous, with, 
    \begin{equation}
        |f(\ket{\psi}) - f(\ket{\phi})| \leq 4 \| \ket{\psi} - \ket{\phi} \|_2
    \end{equation}
\end{lemma}
\textit{Proof.} To prove the above, we first add and subtract the cross term $\operatorname{Tr} \big[\ketbra{\phi}{\phi} \Phi( \ketbra{\psi}{\psi} )]$  to bring the expression to a nice form. Then, we use matrix H\"older inequality, $\operatorname{Tr(A^\dagger B)} \leq \|A\|_\infty \|B\|_1$ to derive an upper bound in terms of Schatten-1 norm (or trace distance),  
\begin{align}
    |f(\ket{\psi})\!& -\! f(\ket{\phi})|  
    \nonumber\\&= \operatorname{Tr} \big[\ketbra{\psi}{\psi} \Phi\left( \ketbra{\psi}{\psi}  \right)\big]\!-\!\operatorname{Tr} \big[\ketbra{\phi}{\phi} \Phi\left( \ketbra{\phi}{\phi}  \right)\big],  \nonumber\\
    &= \operatorname{Tr} \big[(\ketbra{\psi}{\psi} - \ketbra{\phi}{\phi}) \Phi\left( \ketbra{\psi}{\psi}  \right)\big] \nonumber\\
    &\quad - \operatorname{Tr} \big[\ketbra{\phi}{\phi} \Phi\left( \ketbra{\psi}{\psi} - \ketbra{\phi}{\phi}  \right)\big], \nonumber \\
    &\leq \|\ketbra{\psi}{\psi} - \ketbra{\phi}{\phi} \|_1 + \|\Phi\left( \ketbra{\psi}{\psi} - \ketbra{\phi}{\phi} \right)\|_1,  \nonumber\\
    &\leq 2 \|\ketbra{\psi}{\psi} - \ketbra{\phi}{\phi} \|_1,  \nonumber\\
    \label{Apeq:Lip_e}
    &\leq 4 \| \ket{\psi} - \ket{\phi} \|_2.
\end{align}
In the 5-th line, we use the data processing inequality, by noting that the trace distance $2 T(\rho,\sigma) = \| \rho - \sigma \|_1$. Finally, in the 6-th line  we have used the arguments in Ref.~\cite{Mele2024introductiontohaar} (section 9.2, Example 54), that show that $ \|\ketbra{\psi}{\psi} - \ketbra{\phi}{\phi} \|_1 \leq  2 \| \ket{\psi} - \ket{\phi} \|_2$.

\subsection{MCMC Details}\label{App:MCMC}
 In this section of the Supplemental Material, we describe the Markov Chain Monte Carlo (MCMC) sampling procedure used for DFE, \Eq{eq:def_fidelity}, of arbitrary states in the stabilizer scar subspace. Ref.~\cite{flammia2011direct} guarantees that, if the relevant expectation values are known exactly, the number of MCMC samples required is a constant that depends only on the desired precision $\epsilon$ and failure probability $\delta$. If these expectation values are instead estimated experimentally from a finite number of measurement shots, then, following Ref.~\cite{leone2023nonstabilizerness} and the arguments given in the main text, the number of copies of the quantum state required additionally scales polynomially with the number of qubits, i.e., the protocol is efficient.

While Refs.~\cite{flammia2011direct,leone2023nonstabilizerness} establish theoretical guarantees for the DFE procedure, they assume access to i.i.d. samples from the target state Pauli distribution $P_{\rho}(k)$. Generating such samples is not trivial since there are exponentially many Paulis to sample from. In this manuscript, we use a Metropolis-Hastings (MH) algorithm that proposes Pauli operators with non-zero support on $\mathcal{S}$ with \textit{uniform probability}, clearly not an efficient proposal scheme. Alternative schemes will be presented in a forthcoming publication. 
\begin{figure}[t]
    \centering
    \includegraphics[width=0.7\linewidth]{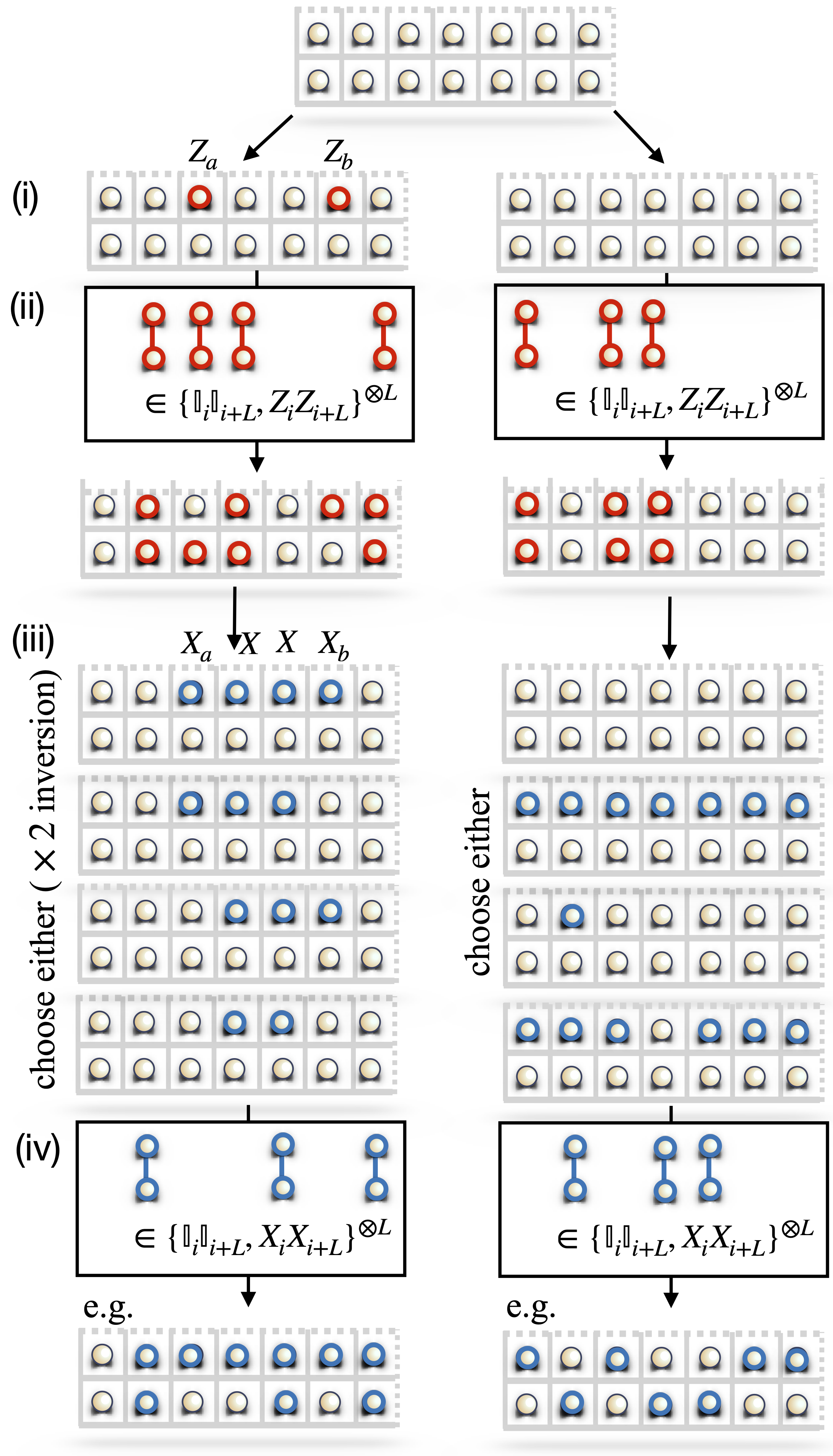}
    \caption{\textit{MCMC suggest-step decision tree.} Shown is the decision tree used to uniformaly generate Pauli operators with nonzero expectation value within the stabilizer scar subspace. This proceeds in four steps: (i) a $Z$ operator pattern is generated, in which either a ``defect hops'' (left column) or does not (right column); (ii) this $Z$ pattern is further dressed by a random product of $\{ \mathbb{I}_{i}\mathbb{I}_{i+L}, Z_{i} Z_{i+L}\}^{\otimes L}$ operators; (iii) based on the resulting $Z$ pattern, an $X$ pattern is generated that either connects the defect's origin and target (left column, not shown are the inverted $X$ patterns in the top row) or merely changes its flavor (right column); and (iv) this $X$ pattern is similarly dressed by a random product of $\{ \mathbb{I}_{i}\mathbb{I}_{i+L}, X_{i} X_{i+L}\}^{\otimes L}$ operators. The corresponding matrix elements, together with the phase obtained from \Eq{eq:formofpaulis}, are tracked throughout to construct a $4L \times 4L$ matrix in the QMBS subspace. }
    \label{fig:decisiontree}
\end{figure}

\subsubsection{Metropolis-Hasting algorithm}
The algorithm proposes only Pauli operators that have non-zero matrix elements within the target scar subspace, i.e., operators drawn from the set 
\begin{equation}
    P_{\mathcal{S}} := \{ W_k  \in \mathcal{P}_n\, | \, \exists \rho \in \mathcal{S} \text { with } \text{Tr}[\rho W_k] \neq 0  \}
\end{equation}
For the prototype model of \Eqs{eq:Hamiltonian}{eq:ss1}, as we show in Eq.~\eqref{eq:countingPS} below, this set contains $|P_{\mathcal{S}}| = (n^2 - n + 2)d$ operators,  still exponentially many in $n$, but a small fraction of the full set of $d^2$ Pauli operators. 

To ensure that only Pauli operators from the set $P_{\mathcal{S}}$ are proposed, we use a decision tree, which is summarized compactly in \Fig{fig:decisiontree}. We seek Pauli operators with non-zero matrix elements between pairs of scar basis states. It is useful to characterize each basis state in terms of a ``defect'' located at site $k=1,\dots,L$ with one of four flavors $a=1,2,3,4$, corresponding to the four panels of \Fig{fig:model}(b). The allowed Pauli operators then either move the defect and/or change its flavor, act as a diagonal operator. They are parameterized as
\begin{align}\label{eq:formofpaulis}
    W = (-i)^{\mathbf{a}\cdot \mathbf{b}}X^{a_1}\dots X^{a_{2L}} Z^{b_1}\dots Z^{b_{2L}}\,,
\end{align}
where $\mathbf{a},\mathbf{b} \subset \{0,1\}^{2L}$. In the first step, \Fig{fig:decisiontree}(i), the pattern of $Z$ operators is determined: with probability $p=\frac{2+2L}{2+2L+4L(L-1)}$ no ``single'' $Z$ operators are generated, while with probability $1-p$ two ``single'' $Z$ operators are placed at two randomly chosen, distinct locations $i \in \{1,\dots,L\}$ (the ``upper row'' of the $L\times 2$ lattice in \Fig{fig:model}(b)), corresponding to a ``hopping'' of a defect. The operator $W$ is then ``dressed'' by multiplication with a random product of $\{ \mathbb{I}_{i}\mathbb{I}_{i+L}, Z_{i} Z_{i+L}\}^{\otimes L}$, where the number $k\in\{0,\dots,L\}$ of $Z_{i}Z_{i+L}$ stabilizer factors included is drawn from a binomial distribution, i.e., there are $\binom{L}{k}$ ways of placing $k$ such stabilizer operators, see \Fig{fig:decisiontree}(ii).

This is followed by the placement of $X$ operators in \Fig{fig:decisiontree}(iii): If no single $Z$ operator was placed in the first step (right column),  then the defect does not hop, but its flavor may still change. In this case, one of four $X$ patterns is chosen:
\begin{subequations}
\begin{align}
    \mathbb{I} &\qquad \text{with probability } \frac{1}{2L+2}, \\
    \prod_{i=1}^{L} X_i &\qquad \text{with probability } \frac{1}{2L+2}, \\
    X_k &\qquad \text{with probability } \frac{L}{2L+2}, \\
    \prod_{\substack{i=1 \\ i\neq k}}^{L} X_i &\qquad \text{with probability } \frac{L}{2L+2},
\end{align}
\end{subequations}
where, in the last two cases, the location $k=1,\dots,L$ is chosen with equal probability. If, on the other hand, two ``single'' $Z$ operators were placed in the first step (left column), the $X$ pattern is largely fixed: it must accommodate the ``hopping'' of the defect, although the flavor may still change in this process. In general, this requires placing an $X$ string connecting the two ``single'' $Z$ locations, which correspond to the origin and target of the defect hop. Owing to the periodic boundary conditions, this can be done in two ways; in addition, the $X$ string may or may not include either of the two ``single'' $Z$ locations, with each choice corresponding to a different combination of origin/target flavor change. This yields $8$ possibilities in total. Finally in \Fig{fig:decisiontree}(iv), as before, the resulting operator is ``dressed'' by multiplication with a random product of $\{ \mathbb{I}_{i}\mathbb{I}_{i+L}, X_{i}X_{i+L}\}^{\otimes L}$.
As the decision tree is traversed, the matrix elements between (sets of) basis states are tracked, the phase appearing in \Eq{eq:formofpaulis} is computed, and a corresponding $4L \times 4L$ matrix is constructed, corresponding to the projection of $W$ into the scar subspace.

This prescription allows one to count the number of nonzero Pauli observables supported on the subspace by summing over all branches of the decision tree in \Fig{fig:decisiontree}. There are $\binom{L}{2}=L(L-1)/2$ ways of placing the single-$Z$ operators, each admitting $8$ corresponding $X$-operator patterns (left column of the figure), together with $2+2L$ configurations containing only flavor-changing $X$ patterns. Each such pattern can further be dressed independently by $2^L$ products from $\{\mathbb{I}_{i}\mathbb{I}_{i+L},\,X_iX_{i+L}\}^{\otimes L}$ and $2^L$ products from $\{\mathbb{I}_{i}\mathbb{I}_{i+L},\,Z_iZ_{i+L}\}^{\otimes L}$, yielding
\begin{align}\label{eq:countingPS}
|P_{\mathcal{S}}|
&=2^{2L}\left(8\binom{L}{2}+2+2L\right)
=d\left(n^2-n+2\right).
\end{align}

The proposed $t$-th sample $k_{t+1}^{\rm prop.}$ is accepted or rejected based on the ratio~\cite{brooks2011handbook}
\begin{align}
    \alpha = \frac{P_{\rho}(k^{\rm prop.}_{t+1})}{P_{\rho}(k_{t})}\,,\quad  P_{\rho}(k)=\chi_{\rho}^2(k)= \frac{\text{Tr}[\rho P_{k}]^2}{{d}} 
\end{align}
where $k_t$ is the current sample in the MCMC. Drawing a uniform random number $u\in[0,1]$, the proposal is accepted, $k_{t+1}=k_{t+1}^{\rm prop.}$, if $\alpha > u$, and rejected, $k_{t+1}=k_{t}$, otherwise (i.e., if $\alpha \le u$).

\subsubsection{I.i.d. MCMC samples}
As is inherent to MCMC methods, successive samples are generally autocorrelated because each sample is generated from the previous one \cite{robert2009metropolis}. To quantify this effect, we compute the Pearson correlation coefficient of the sampled estimators $E_\tau\equiv \chi_\sigma(k_\tau)/\chi_\rho(k_\tau)$. Specifically, the autocorrelation at lag $\tau$ is defined as the ratio of the covariance between samples separated by $\tau$ to the variance of the series, 
\begin{align}
        C_\tau \equiv \frac{\text{Cov}(E_t,E_{t+\tau})}{\text{Var}(E_t)}=\frac{\sum_{t=1}^{N-\tau} (E_t - \bar{E})(E_{t+\tau}  - \bar{E})}{\sum_{t=1}^{N} (E_t - \bar{E})^2}\,,
    \label{eq:autocorrelation}
\end{align}
where $\bar{E}$ is the sample mean and $N$ is the total number of samples in the chain. The autocorrelation determines the effective sample size $N_{\text{eff}}$, which is generally smaller than the total number of samples $N$, thereby determining how fast the sample mean converges to the true value. For our numerical simulations we discard the first $1000$ samples as a burn-in period and thin the chain, keeping only every $100$-th sample thereafter. We find that the auto-correlation factor $C_{100} \approx0$, indicating a negligible residual autocorrelation. Consequently,  we treat our samples as being generated independently, i.e, $N_{\text{eff}} \approx N$. 

When the autocorrelation cannot be made negligible through thinning, the effective sample size can be estimated analytically. Concretely, consider the random variable $\bar{E} = \frac{1}{N} \sum_i E_i$, where $E_i$ is the random variable associated with the MCMC sampling, then
\begin{align}
    \text{Var}(\bar{E})= \frac{\text{Var}(E_t)}{N^2}\Big( N + 2 \sum_{t=1}^N \sum_{\tau=1}^{N-t}C_\tau\Big)\,,
\end{align}
which recovers the familiar $1/N$ scaling if the second term vanishes. One can write
\begin{align}
    \sum_{t=1}^N \sum_{\tau=1}^{N-t}C_\tau = \sum_{\tau=1}^{N} \sum_{t=1}^{N-\tau} C_\tau  = \sum_{\tau=1}^N (N-\tau) C_\tau\,,
\end{align}
and, parameterizing $C_\tau\equiv\epsilon_{\rm acf}e^{-k (\tau - 1)}$ where $\epsilon_{\rm acf}= C_1$, and $k$ a positive constant which can be extracted from numerics, one gets
\begin{align}
    \text{Var}(\bar{E})= \frac{\text{Var}(E_t)}{N} \Big(1 + \frac{2 \epsilon_{\rm acf}}{N}\sum_{\tau=1}^N (N-\tau) e^{-k\tau} \Big)\,.
\end{align}

\begin{figure*}[t]
    \centering
    \includegraphics[width=1.0\linewidth]{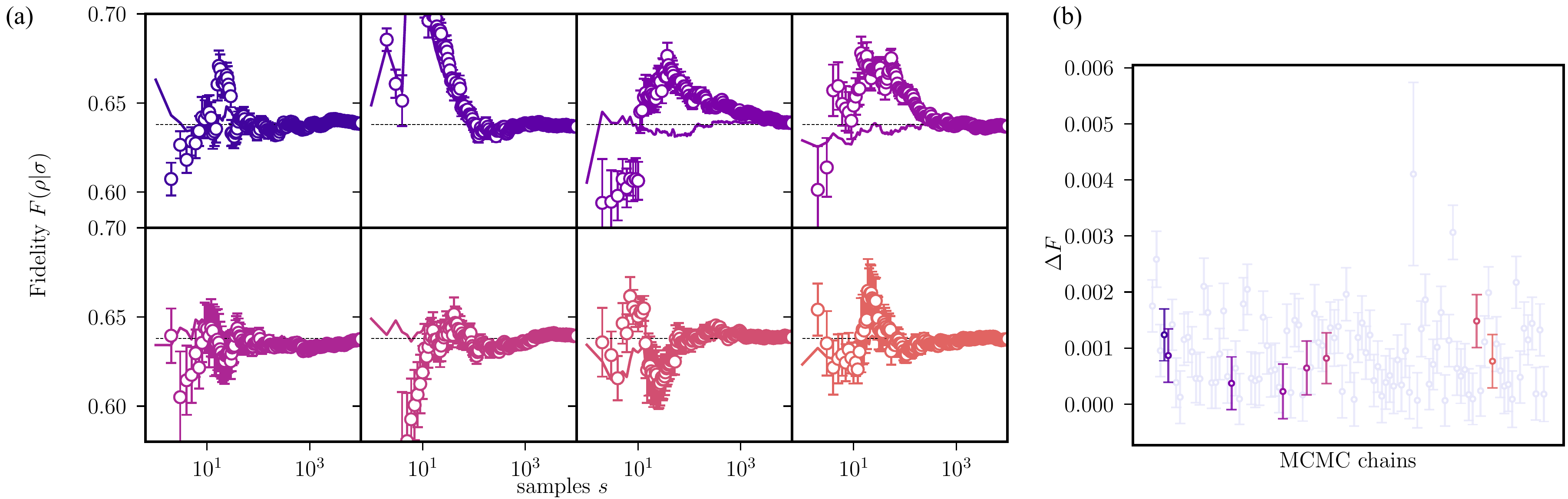}
    \caption{\textit{Additional noisy circuit simulation with shots.} 
(a) Results for $n=10$ qubits, $\ell=15$ circuit layers, and depolarizing noise strength $p=0.001$, analogous to \Fig{fig:circuit}(c) of the main text. Eight different randomly selected MCMC chains are shown out of $100$ generated. Circles denote shot-sampled estimates ($10^4$ shots per Pauli observable), while solid lines show the corresponding exact Pauli expectation values. Some shot-sampled trajectories exhibit an apparent systematic offset from the exact curves. This arises from the construction of the statistical uncertainty: the error bars represent the accumulated uncertainty of the fidelity estimate after $s$ MCMC samples, obtained via error propagation. As a result, an early (or intermediate) statistical outlier in a Pauli expectation value can bias the cumulative estimate until its effect is gradually averaged out by subsequent samples.
(b) Same  as the bottom panel of \Fig{fig:circuit}(c) of the main text, showing the final fidelity precision after $s=10^4$ MCMC samples.
    \label{fig:differntMCMCruns}}
\end{figure*}

\subsubsection{Simulating an experiment with shots}
To demonstrate that the number of shots required to achieve the desired precision falls below the bound in Eq.~\eqref{eq:leonebound}  we classically emulate the shot-based measurements of expectation values for a test state $\sigma$ and target state $ \rho  \in \mathcal{S}$. Recall that $\rho$ is the ideal state for which expectation values are known exactly, while  $\sigma$ is the state that is estimated in experiment with a-priori unknown expectation values. The following analysis requires knowledge of the exact expectation value $\operatorname{Tr}(W_i \sigma)$, which is inaccessible in practice; it can be understood, in a sense, as the ideal (i.e. lower bound) number of shots, while the results in \Eq{eq:leonebound} of the main text give a  worse-case guarantee (upper bound).  Given the expectation
value $\langle W_{k_i} \rangle_\sigma \equiv \mathrm{Tr}[\sigma W_{k_i}]$, we define the outcome probabilities
\begin{align}\label{eq:binomdistr}
    p^{(i)}_{\pm} = \frac{1 \pm \langle W_{k_i} \rangle_\sigma}{2}\,,
\end{align}
which satisfy $p^{(i)}_+ + p^{(i)}_- = 1$. Sampling from this Bernoulli
distribution yields measurement outcomes corresponding to projective measurements
of $W_{k_i}$.

To determine the number of shots $m_i$ required to achieve a prescribed precision
and confidence, we demand
\begin{align}\label{eq:reqindivudual}
    \text{Pr}\!\left(|\hat{Y} - Y| \ge \epsilon\right) \le \delta\,,
\end{align}
where $Y = \frac{1}{s}\sum_{i=1}^{s} X_i$ and
$\hat{Y} = \frac{1}{s}\sum_{i=1}^{s} \hat{X}_i$ are the estimators built from
exact and shot-sampled expectation values, respectively, and $X_i$, $\hat{X}_i$
are defined below \Eq{eq:def_fidelity}. Suppose that for each Pauli
sample $i$ we have $ \text{Pr}(|\hat{X}_i - X_i| \ge \epsilon_i) \le \delta_i$.
Then, by the triangle inequality and independence,
\begin{align}
    \text{Pr}\!\left(\frac{1}{s}\sum_{i=1}^{s}|\hat{X}_i - X_i|
    \le \frac{1}{s}\sum_{i=1}^{s}\epsilon_i\right)
    \ge \prod_{i=1}^{s}(1 - \delta_i)\,.
\end{align}
Hence, Eq.~\eqref{eq:reqindivudual} is satisfied provided each term obeys
\begin{align}\label{eq:individual_bound}
     \text{Pr}\!\left(|\hat{X}_i - X_i| \ge \epsilon\right)
    \le 1 - (1-\delta)^{1/s}\,.
\end{align}
Following Ref.~\cite{flammia2011direct}, the shot-based estimator reads
\begin{align}
    \hat{X}_i = \frac{1}{m_i\, \langle W_{k_i} \rangle_\rho }
    \sum_{j=1}^{m_i} A_{ij}\,,
\end{align}
where $A_{ij} \in \{+1,-1\}$ are i.i.d.\ draws from the distribution in
Eq.~\eqref{eq:binomdistr} and $\langle W_{k_i} \rangle_\rho \equiv \operatorname{Tr} (W_{k_i} \rho)$. The failure probability for sample $i$ is therefore
\begin{align}\label{eq:shotsrequired}
    & \text{Pr}\left(|\hat{X}_i - X_i| \ge \epsilon\right)
    =  \text{Pr}\Big(\big|\langle W_{k_i} \rangle_\sigma
      - \frac{1}{m_i}\sum_{j=1}^{m_i} A_{ij}\big|
      \ge \epsilon\,\langle W_{k_i} \rangle_\rho \Big) \nonumber\\
    &= F\left(
         \left\lfloor m_i\!\left(p^{(i)}_+ - \tfrac{\epsilon}{2}
         \langle W_{k_i} \rangle_\rho\right)\right\rfloor;
         \,m_i,\,p^{(i)}_+\right) \nonumber\\
    &\quad + 1 - F\left(\left\lceil m_i\!\left(p^{(i)}_+
         + \tfrac{\epsilon}{2}\langle W_{k_i} \rangle_\rho \right)\right\rceil;
         \,m_i,\,p^{(i)}_+\right)\,,
\end{align}
where the binomial cumulative distribution function is,
\begin{align}
    F(a;\,m,p) = \sum_{n=0}^{a}\binom{m}{n}p^{n}(1-p)^{m-n}\,.
\end{align}
While the expression in Eq.~\eqref{eq:shotsrequired} is not strictly monotonic in $n$, it is monotonic in  aggregate. Hence, the required number of shots $m_i$ for each expectation value
$\langle W_{k_i}\rangle$ to achieve precision $\epsilon$ and failure probability
$\delta$ is then obtained by finding the smallest $m_i$ such that
moving average of Eq.~\eqref{eq:shotsrequired} does not exceed the right-hand side of
Eq.~\eqref{eq:individual_bound}.

\begin{figure*}[t]
    \centering
    \includegraphics[width=0.79\linewidth]{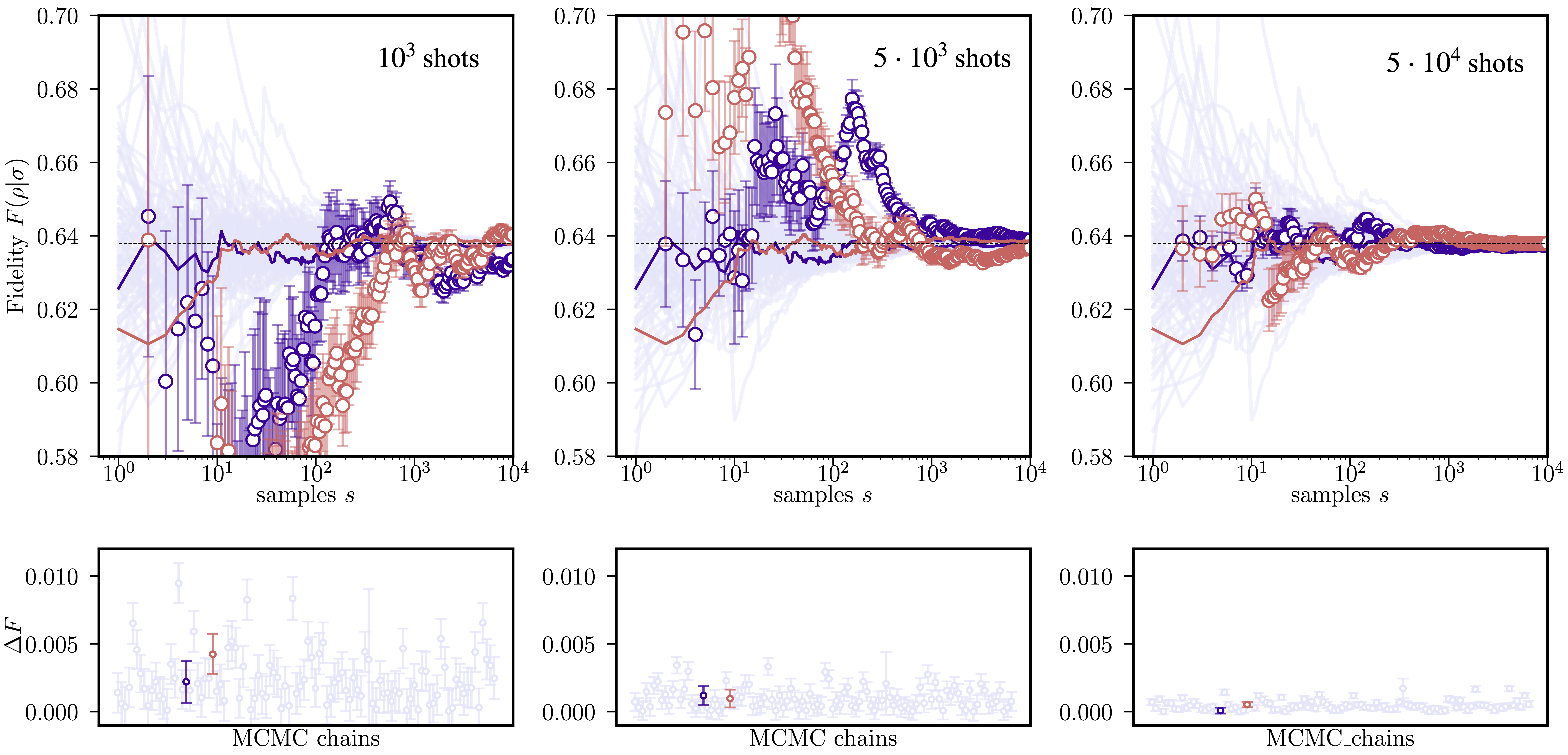}
\caption{\textit{Shot dependence.} The same MCMC chain as in \Fig{fig:MCMC}(c) of the main text applied to noisy circuit simulations with $n=10$ qubits, $\ell=15$ layers, and depolarizing noise strength $p=10^{-3}$. The number of shots per Pauli observable is varied between $10^{3}$ (10\% of the shots used in the main text, left panel), $5\cdot 10^{3}$ (50\%, middle panel), and $5\cdot 10^{4}$ (500\%, right panel). The bottom panels show the resulting precision after $s=10^4$ MCMC samples.
\label{fig:differntMCMCshotsrunsshots}}
\end{figure*}

Additional numerical studies complementing \Fig{fig:circuit}(c) of the main text are presented in \Fig{fig:differntMCMCruns} ($n=10$ qubits, $\ell=15$ layers, $p=10^{-3}$). In (a), eight \textit{different} randomly selected MCMC chains are shown, exhibiting behavior consistent with the results presented in the main text. For some chains, an apparent systematic deviation is observed at small $s$ between the estimates obtained from exact Pauli expectation values (solid lines) and those obtained from shot-sampled measurements (circles). This behavior originates from the way the statistical uncertainty is computed: the shown error bars correspond to the accumulated uncertainty of the fidelity estimate after $s$ MCMC samples, represented using the sample variance obtained via error propagation, and uncertainties at successive data points are not independent. Consequently, an atypical shot-based estimate, i.e., an outlier, at an early stage of the sampling process can bias the cumulative estimate for a number of subsequent MCMC samples, with the resulting deviation persisting until it is averaged out by additional samples. Ultimately, all trajectories converge, achieving the improved precision and failure probability that we observe is significantly better than the bounds derived in the main text. 

\Fig{fig:differntMCMCshotsrunsshots} shows the same MCMC chains as in 
\Fig{fig:circuit}(c) of the main text. However, we vary the shots used to evaluate each Pauli from  $10^{3}$ (10\% of the shots used in the main text, left panel), $5\cdot 10^{3}$ (50\%, middle panel), and $5\cdot 10^{4}$ (500\%, right panel). These results are also consistent with the bounds presented in the main text.

\end{document}